\documentclass[a4paper,UKenglish,cleveref, autoref]{tgdk-v2021}

\usepackage{alltt}
\usepackage{amsmath,amssymb,amsfonts}
\usepackage{booktabs}
\usepackage{fancyvrb}
\usepackage{ifthen}
\usepackage[utf8]{inputenc}
\usepackage{listings}
\usepackage{lmodern}
\usepackage{relsize}
\usepackage[group-separator={,}]{siunitx}
\usepackage{tabto}
\usepackage{textcomp} 
\usepackage{url}
\usepackage{verbatim}

\newcommand{\GClassic}[3]{\ensuremath{
  \ifthenelse{\equal{#1}{}}{\emptyset}{\{#1\}},
  \ifthenelse{\equal{#2}{}}{\emptyset}{\{#2\}},
  \{#3\}
}}
\newcommand{\GQuantifierSet}[3]{\ensuremath{
  \ifthenelse{\equal{#1}{}}{}{\forall\{#1\},}
  \ifthenelse{\equal{#2}{}}{}{\exists\{#2\},}
  \{#3\}
}}
\newcommand{\GQuantifierSetNoComma}[3]{\ensuremath{
  \ifthenelse{\equal{#1}{}}{}{\forall\{#1\}\,\allowbreak}
  \ifthenelse{\equal{#2}{}}{}{\exists\{#2\}\,\allowbreak}
  \{#3\}
}}
\newcommand{\GQuantifier}[3]{\ensuremath{
  \ifthenelse{\equal{#1}{}}{}{\forall#1}
  \ifthenelse{\equal{#2}{}}{}{\exists#2}
  \{#3\}
}}
\newcommand{\GIndex}[3]{\ensuremath{
  \ifthenelse{\equal{#1}{}}{}{^{\forall\{#1\}}}
  \ifthenelse{\equal{#2}{}}{}{_{\exists\{#2\}}}
  \{#3\}
}}
\newcommand{\GChosen}[3]{\GQuantifierSetNoComma{#1}{#2}{#3}}
\newcommand{\G}[3]{\ensuremath{\GChosen{#1}{#2}{#3}}}
\newcommand{\GTW}[1]{\ensuremath{\underline{#1}}}
\newcommand{\GT}[3]{\GTW{\GChosen{#1}{#2}{#3}}}

\DeclareUnicodeCharacter{AC}{\ensuremath{\neg}}
\DeclareUnicodeCharacter{D7}{\ensuremath{\times}}
\DeclareUnicodeCharacter{3B5}{\ensuremath{\epsilon}}
\DeclareUnicodeCharacter{202F}{\ensuremath{~}}
\DeclareUnicodeCharacter{2115}{\ensuremath{\bf{N}}} 
\DeclareUnicodeCharacter{2192}{\ensuremath{\rightarrow}}
\DeclareUnicodeCharacter{2194}{\ensuremath{\leftrightarrow}}
\DeclareUnicodeCharacter{2200}{\ensuremath{\forall}}
\DeclareUnicodeCharacter{2203}{\ensuremath{\exists}}
\DeclareUnicodeCharacter{2205}{\ensuremath{\emptyset}}
\DeclareUnicodeCharacter{2208}{\ensuremath{\in}}
\DeclareUnicodeCharacter{2209}{\ensuremath{\not\in}}
\DeclareUnicodeCharacter{2218}{\ensuremath{\circ}}
\DeclareUnicodeCharacter{2227}{\ensuremath{\wedge}}
\DeclareUnicodeCharacter{2228}{\ensuremath{\vee}}
\DeclareUnicodeCharacter{2229}{\ensuremath{\cap}}
\DeclareUnicodeCharacter{222A}{\ensuremath{\cup}}
\DeclareUnicodeCharacter{2260}{\ensuremath{\not=}}
\DeclareUnicodeCharacter{2264}{\ensuremath{\leq}}
\DeclareUnicodeCharacter{2265}{\ensuremath{\geq}}
\DeclareUnicodeCharacter{2286}{\ensuremath{\subseteq}}
\DeclareUnicodeCharacter{2286}{\ensuremath{\subseteq}}
\DeclareUnicodeCharacter{2287}{\ensuremath{\supseteq}}
\DeclareUnicodeCharacter{228A}{\ensuremath{\subsetneq}}
\DeclareUnicodeCharacter{22A6}{\ensuremath{\vdash}} 
\DeclareUnicodeCharacter{22A7}{\ensuremath{\models}}
\DeclareUnicodeCharacter{22A8}{\ensuremath{\vDash}}
\DeclareUnicodeCharacter{22C3}{\ensuremath{\cup}}
\DeclareUnicodeCharacter{2BA3}{\ensuremath{\nearrow}}
\DeclareUnicodeCharacter{0394}{\ensuremath{\Delta}}

\usepackage{graphicx}
\usepackage[dvipsnames,svgnames]{xcolor} 

\usepackage[normalem]{ulem} 

\makeatletter
\font\uwavefont=lasyb10 scaled 700
\def\spelling{\bgroup\markoverwith{\lower3.5\p@\hbox{\uwavefont\textcolor{Red}{\char58}}}\ULon}
\def\grammar{\bgroup\markoverwith{\lower3.5\p@\hbox{\uwavefont\textcolor{LimeGreen}{\char58}}}\ULon}
\def\phrasing{\bgroup\markoverwith{\lower3.5\p@\hbox{\uwavefont\textcolor{RoyalBlue}{\char58}}}\ULon}

\newcommand\remove{\bgroup\markoverwith{\textcolor{red}{\rule[0.5ex]{2pt}{0.4pt}}}\ULon}
\makeatother

\title{Notation3: Graph Terms for RDF} 

\author{Dörthe {Arndt}}{TU Dresden, Germany \and ScaDS.AI Dresden/Leipzig, Germany \and \url{https://iccl.inf.tu-dresden.de/web/D\%C3\%B6rthe_Arndt
}}{doerthe.arndt@tu-dresden.de}{https://orcid.org/0000-0002-7401-8487}{
}

\author{Pierre-Antoine {Champin}}{W3C \and Inria, University Côte d’Azur, CNRS, I3S UMR 7271, France \and \url{https://champin.net/}}{pierre-antoine@w3.org}{https://orcid.org/0000-0001-7046-4474}{
}

\authorrunning{D. Arndt and P.-A. Champin} 

\Copyright{Jane Open Access and Joan R. Public} 

\ccsdesc[500]{Information systems~Resource Description Framework (RDF)}
\ccsdesc[300]{Theory of computation~Logic and verification}
\ccsdesc[300]{Theory of computation~Semantics and reasoning} 

\keywords{Notation3 Logic, Graph terms, Semantics, Rules, RDF} 

\category{} 

\relatedversion{} 

\supplement{The authors declare that this article involves no relevant supplementary resources.} 

\acknowledgements{The authors would like to thank all the participants of the W3C Notation 3 Community Group\footnote{\url{https://www.w3.org/community/n3-dev/}} for fruitful discussions and their feedback on early stages of this work.}

\nolinenumbers 

\Volume{1}
\Issue{1}
\Article{42}
\DateSubmission{Date of submission}
\DateAcceptance{Date of acceptance}
\DatePublished{Date of publishing}
\SectionAreaEditor{TGDK section area editor}

\begin{document}

\maketitle

	\begin{abstract}
 Notation3 Logic was introduced as an extension of RDF allowing the user to express rules on top of triples and to make statements about graphs. The original proposal of  the logic only came with an informal specification of its semantics which 
has been refined later on. While this refinement clarified the meaning of rules and implicitly quantified variables, it left the meaning of graphs used as terms -- so-called \emph{graph terms} -- open. In this paper we aim to close this gap. We discuss in detail, how graph terms are used in Notation3, which formal requirements they need to meet, and provide examples to illustrate the difficulties when defining their semantics. We then introduce an abstract syntax for the logic and define its semantics in two steps: similarly to RDF's simple interpretations, we define interpretations and basic entailment for Notation3 which covers plain RDF triples but also triples containing graph terms. N3-entailment is then an extension of basic entailment which assigns meaning to the different special predicates of the logic like for example the implication symbol. N3-entailment is defined in a way that it can be extended whenever a new predicate is added to the logic. As Notation3 Logic is currently under development, we expect many predicates to be added. Our research provides and important step towards the full standardisation of the logic.

\end{abstract}

\section{Introduction}

Ever since the invention of the Semantic Web, the community saw the need to allow users to make statements about statements. Different approaches like singleton properties~\cite{nguyen2014don},
the inclusion of named graphs to RDF~\cite[Section 4]{rdf11-concepts}, or the integration of RDF-star~(formally RDF*)~\cite{hartig2017foundations,hartig2021rdf} into the upcoming RDF 1.2, are only some indicators for it. And while these proposals and standards to some extend address the problem, they all have some draw-backs: singleton properties and RDF-star only support statements about triples but not about RDF graphs and for named graphs, there is so far no agreement about their semantics~\cite{namedgraphs}. 

In this paper we focus on another framework which supports statements about statements: Notation3 Logic~(N3)~\cite{N3Logic}. N3 was proposed by Tim Berners-Lee and others to provide a simple syntax for RDF -- this resulted in RDF-Turtle~\cite{Prudhommeaux2014RT} -- but also as a logical extension for RDF. N3 supports so-called graph terms which are conceptually very similar to RDF-star's triple terms with the only difference that N3 does not only allow triples as terms but entire graphs. 
Of course, this comes with challenges: graph terms could stand for themselves or for bigger graphs, scoping of blank nodes is far from trivial, and graph terms are used to express logical constructs. When originally proposed, the semantics of N3 has only been defined informally and different implementations like Cwm~\cite{cwm}, EYE~\cite{eyepaper} or Data-Fu~\cite{kafer2018rule} were developed to reason on it. 
The theoretical works which followed later on to define a model theoretic semantics for N3 did focus on rules~\cite{arndt2023notation3,exN3} and on implicit quantification
~\cite{arndt2019implicit}, but did not fully cover the semantics of graph terms as it was originally intended.

In this paper we aim to fill this gap: We analyse the requirements a graph term semantics has to comply with according to the informal specifications and the existing implementations. We then propose an abstract syntax for N3 and define its semantics in two steps: we first specify a base semantics which covers the meaning of graph terms but does not assign any meaning to the built-in predicates of the logic. In that sense our base semantics is very similar to RDF's simple entailment. In a second step we define N3-semantics which covers the logical predicates in N3's \texttt{log}-namespace\footnote{\url{http://www.w3.org/2000/10/swap/log}}. Here we mainly focus on logical implication~(\texttt{log:implies}) and on graph inclusion~(\texttt{log:includes}). Especially the former is crucial for N3 as this predicate enables rule-based reasoning. We furthermore show that our definitions are in-line with the requirements.

The remainder of this paper is structured as follows: In Section~\ref{motivation} we analyse how graph terms are used in N3, and we derive the requirements for their semantics in Section~\ref{sec:req}. In Section~\ref{sec:syntax} we define an abstract syntax for N3 and explain how it maps to the concrete serialization syntax. In Section~\ref{semantics} we then provide a formal semantics for the abstract syntax, and we show in Section~\ref{sec:eval} that the semantics of graph terms meets our requirements. In Section~\ref{sec:relatedwork} we discuss the related work and explain how our formalisation relates to other approaches to enable statements about statements in RDF. With Section~\ref{sec:conclusion} we conclude our paper.

\section{Graph terms in N3}\label{motivation}
Before formally introducing N3 and our proposal for a graph term semantics, we give a few examples to make the reader familiar with the logic. N3 extends the Turtle syntax for RDF~\cite{beckett_turtle_nodate}, which means that each RDF  turtle triple is also an N3 triple.
We can, for example, state that\footnote{We omit the prefix definitions where not relevant to keep the paper short.}
\begin{equation}\label{eq:triple}
\texttt{:alice :likes :pasta.}
\end{equation}
This triple informally means that \emph{``Alice likes pasta.''} Just as in RDF, we have literals and IRIs which refer to resources, and blank nodes which  are used as existentially quantified variables. N3 furthermore supports graph terms. We can, for example, state that
\begin{equation}\label{eq:firstGT}
\texttt{:bob :says \{:alice :likes :pizza, :pasta\}. }
\end{equation}
Which informally means that \emph{``Bob says \textbf{that} Alice likes pizza \textbf{and} pasta.''} Note that, while RDF-star only allows a single triple to be used as a term, here the object is a set of triples which we interpret as a conjunction. This makes it possible to have co-referring blank nodes in graph terms. We could modify the previous statement and get:
\begin{equation}\label{eq:bob2}
\texttt{:bob :says \{\_:x :likes :pizza. \_:x :likes :pasta\}. }
\end{equation}
Despite being a little bit vague, we could give an informal interpretation before. In this case it is more complex because there is implicit existential quantification involved. In our proposal, the above means \emph{``Bob says that \textbf{there exists} someone who likes pizza and pasta''}. 
The blank node is existentially quantified \emph{within} the graph term and not outside of it\footnote{
  In the latter case, Formula~\ref{eq:bob2} would instead mean
 \emph{``There exists someone of whom Bob says they like pizza and pasta''}.
}. It has to be said, that this interpretation is still subject to discussion, which is why a semantics of N3 needs to be flexible to be able to adjust possible future changes in the scoping of implicit quantification. 


Graph terms are not only used to make statements about statements but also to express logical operations. A rule can be written as follows:
\begin{equation}\label{eq:rule1}
\texttt{\{?x :likes :pasta \} => \{?x :likes :spaghetti \}.}
\end{equation}
In this rule, the \texttt{?x} stands for a~(implicitly) universally quantified variable and the sign \texttt{=>} is syntactic sugar for the N3 predicate \texttt{log:implies}. The above rule informally means that \emph{``If someone likes pasta, then they also like spaghetti.''}. We see that the rule above is simply a triple with graph terms as subject and object and a special predicate \texttt{log:implies}. This is a design choice of N3: All data and axioms in N3 are expressed through triples, and all semantic decisions~(e.g., scoping of variables) we take for graph terms in statements like Formula~\ref{eq:firstGT} also apply to the graph terms occurring in rules like Formula~\ref{eq:rule1} above. 

We illustrate this by providing the following example:
\begin{equation}\label{eq:weirdrule}
 \texttt{\{:bob :says ?x \} => ?x. }
\end{equation}
This can be freely interpreted as \emph{``Everything Bob says is true.''} 
In order to use this formula, apply it to Formula~\ref{eq:firstGT} and derive the triples 
\begin{equation} \label{eq:res}
\texttt{:alice :likes :pasta, :pizza.}
\end{equation}
\sloppy
we need to consider the dual nature of Formula~\ref{eq:weirdrule}. 
Understanding the IRI \texttt{log:implies} as a predicate, we have an atomic formula~(in first-order-like syntax it would be $\forall x. implies(says(bob,x) ,x)$) and we can unify the universal variable \texttt{?x} with the graph term
\texttt{\{:alice :likes :pasta, :pizza.\}}. As a consequence, we get the formula
\begin{align}
\label{eq:weirdruleinst}
 \texttt{\{:bob :says  \{:alice } & \texttt{:likes pasta, :pizza.\} \}}   \\
 \notag
& \texttt{=> \{:alice :likes :pasta, :pizza.\}. } 
\end{align}
This, in turn, can now be seen as an implication and applied on Formula~\ref{eq:firstGT} it yields Formula~\ref{eq:res}.


\section{Requirements on a formal semantics}\label{sec:req}
The examples in the previous section reveal that there are several difficulties and uncertainties when dealing with graph terms in N3. Some requirements on N3's semantics only become clear after further analysing the logic. Below, we list and explain the main requirements to be taken into account when defining the semantics of N3, with a special focus on graph terms.

\textbf{Semantic compatibility with RDF}
N3 syntactically extends RDF Turtle and each RDF Turtle triple is also an N3 triple. N3 semantics should thus be compatible with RDF semantics. As N3 does not assign special semantics to the different predicates of the RDF and RDFS namespace~(e.g., \texttt{rdf:type}, \texttt{rdfs:domain}), we limit this requirement to compatibility withe RDF simple entailment. N3 should additionally support D-Entailment, but we leave that detail out of the considerations of this paper as our main focus is on graph terms. 

\textbf{Referential Opacity}
One detail which often leads to discussion when a logic supports statements about statements is the question whether or not a constant occurring in a formula term~(as a generalised graph term) refers to the same resource as it does in a stated formula~(e.g.~\cite[Section 6.4.4]{hartig2021rdf}). This is often illustrated by the \emph{superman example} which in N3 syntax would be:
\begin{align}\label{eq:super}
	&\texttt{:superman owl:sameAs :clarkKent.}\\ \label{eq:loislane}
&\texttt{:loisLane :knows \{:superman :can :fly\}}.
\end{align}
Above, we borrowed the predicate \texttt{owl:sameAs} from OWL~\cite{owl}
to express equality. The above means in a free translation: \emph{``Superman is the same as Clark Kent and Lois Lane knows that Superman can fly.''}. Under the assumption that we support equality, we ask the question whether we can derive that 
\begin{equation}\label{eq:clarkKent}
\texttt{:loisLane :knows \{:clarkKent :can :fly\}}
\end{equation}
or in words \emph{``Lois Lane knows that Clark Kent can fly.''}
In case of a positive answer, we say that graph terms are \emph{transparent}, otherwise they are \emph{opaque}.

Opaque graph terms are considered a better design choice for N3. Firstly, as the Superman example above demonstrates, there are use-cases where the derivations entailed by transparency are not desirable. Secondly, if such derivations are desired, it is still possible to enable them by adding a rule of the form:
\begin{equation}\texttt{\{?x owl:sameAs ?y. ?s ?p \{?x ?q ?r\} \}=>\{?s ?p \{?y ?q ?r\} \}.} \label{eq:transparent}
\end{equation}
or a variant thereof. Referential opacity is a weaker, less constraining, option than transparency.

\textbf{Closedness of graph terms}
Another detail to be discussed in the context of graph terms is whether we consider them \emph{closed} or \emph{open}. 
By this we mean whether the term \texttt{\{:clarkKent :can :fly\}} from Formula~\ref{eq:clarkKent} refers to a graph~(a set of triples) containing exactly one triple~(\emph{closed}), namely \texttt{:clarkKent :can :fly.} or to a potentially bigger graph containing \emph{at least} that triple -- but may contain others~(\emph{open}).
Both options have been considered as proposed semantics for datasets in RDF~\cite{namedgraphs}.
%
%
Consider the rule expressed in Formula~\ref{eq:weirdruleinst}.
If the graph term on the left-hand side could denote a graph containing arbitrary extra triples, there would be no way to prove it to be satisfied by the currently known triples, and therefore no way to derive the consequence of that~(or any) rule.
%
Therefore, we consider graph terms to be closed in our semantics.

\textbf{Clarity and flexibility of quantification scopes}
N3 has two kinds of variables: existentially quantified~(or blank nodes, inherited from RDF) and universally quantified~(in order to express rules, as illustrated earlier).
When combined with graph terms, this immediately raises the question of the scope of these quantifications.
Historically, blank nodes~(such as \texttt{\_:y}) have been scoped to the innermost graph containing them~(``locally scoped''), while universal variables~(such as \texttt{?x}) have been scoped to the top-level graph~(``globally scoped'').
While this fits most use-cases, those scoping rules have always been a topic of discussions and a point of divergence among different implementations of N3~\cite{arndt2019implicit}.
For example, they do not tell us precisely how to interpret the following simple graph:
\begin{equation}\label{eq:love}
\texttt{?x :loves \_:y.}
\end{equation}

\noindent We know that both variables are scoped to the same~(top-level) graph, but it is not clear how they interact.
Without additional information, the reader could either interpret the above as (a) 
\emph{``Everybody loves somebody''} or (b) \emph{``There is someone who is loved by everyone''}.

Even if we specify the scoping rule to disambiguate the example above~(and we will, making (a) the correct interpretation), we also need a way to explicitly express the quantification scope of a given variable when the implicit scoping rule are not adequate. Besides the example of interpretation (b) above, consider the example below where we want to describe, in an N3 graph, the content of another N3 document:
\begin{equation}\label{eq:log:semantics}
\texttt{<example.n3>\ log:semantics\ \{\ @forAll\ :x.\ :x\ :loves\ \_:y \}}
\end{equation}

Here, we could not have used the implicitly quantified universal \texttt{?x}, because its scope would have been the whole graph, while what we want to convey is that it is scoped to the graph term representing the content of \emph{example.n3}.
Explicit quantification was also introduced early in N3~(with the keywords \texttt{@forAll} and \texttt{@forSome}), but like the implicit scoping rules, their interpretations and implementations have varied largely.

\section{Syntax}\label{sec:syntax}
Before defining the semantics of N3, we introduce its syntax. 
The concrete syntax of N3, which is used to illustrate examples throughout this paper, is described in the draft report\footnote{\url{https://w3c.github.io/N3/spec/}} of the W3C Notation 3 Community Group. Space limitations prevent from presenting it in details, but the examples provided in Section~\ref{motivation} present its main features, such as graph terms~(Formula~\ref{eq:firstGT}) and rules~(Formula~\ref{eq:rule1}).

To address the last requirement from above, the \emph{flexibility to express different quantification scopes}, but also to have a way to easily deal with quantification,
we define an abstract syntax.  Below, we introduce this abstract syntax, explain its relation to concrete syntax and provide the syntactic definitions needed to define the semantics.

\subsection{Abstract Syntax}

We first define atomic triples, the building blocks of our abstract syntax.

\begin{definition}[Atomic Term]
The set of \emph{atomic terms} consists of
\begin{itemize}
\item
the set $C$ of IRIs,
\item
the set $L$ of Literals,
\item
the set $V$ of Variables.
\end{itemize}
\end{definition}
As we can see, the abstract syntax has only one kind of variable. It does not distinguish between existential~(blank nodes) and universal variables, because contrarily to the concrete syntax, all variables in the abstract syntax will be explicitly quantified.

\begin{definition}[Term]
A \emph{term} can be
\begin{itemize}
\item
an {atomic term},
\item
a {graph term} \GTW{G} where $G$ is a finite
{graph}~(see below).
\end{itemize}
\end{definition}
Note that we make a clear difference between a \emph{graph} and a \emph{graph term}. This difference will become important in Section~\ref{semantics}.

\begin{definition}[Triple]
	A \emph{triple} is a 3-tuple of terms, respectively called the {subject}, {predicate} and {object} of the triple.
\end{definition}
Note that, in contrast to RDF, N3 allows for any kind of term in all positions of a triple.

\vspace{0.5em}
While in RDF a graph is simply defined as a set of triples,
N3 needs a way to specify which variables are universally or existentially quantified, and in which scope.

\begin{definition}[Graph]
An \emph{N3 graph}~(or simply graph) is a tuple $(U, E, F)$ where
\begin{itemize}
\item
$U \subseteq V$ is a~(possibly empty) set of
{universally quantified variables},
\item
$E \subseteq V$ is a~(possibly empty) set of
{existentially quantified variables},
\item
$U$ and $E$ are disjoint,
\item
$F$ is a~(possibly empty) set of {triples}.
\end{itemize}

A {graph} $G=(U, E, F)$ is finite if the sets $U$, $E$ and $F$
are finite.
\end{definition}
%
Note that this model is a strict superset of RDF: any RDF graph can be seen as an N3 graph, where $F$ is the set of triples consituting the graph, $U=\emptyset$, and $E$ is the set of all blank nodes used in $F$. Of course, an RDF graph will use no graph term in $F$.

\begin{remark*}[Notation]
Consider a graph $G=(U, E, F)$.
To improve readability, in our examples, we will use the quantifier $\forall$ to mark the set of universally quantified variables and the quantifier $\exists$ to mark the set of existentially quantified variables, and ommit those sets when they are empty.
For example, $(\{x\}, \{y\}, \{(x, \text{:likes}, y)\})$ will be noted \G{x}{y}{(x, \text{:likes}, y)},
$(\{x\}, \{\}, \{(x, \text{:likes}, \text{:pizza})\})$ will be noted \G{x}{}{(x, \text{:likes}, \text{:pizza})},
 and $(\{\}, \{\}, \{(\text{:alice}, \text{:likes}, \text{:pizza})\})$ will be noted \G{}{}{(\text{:alice}, \text{:likes}, \text{:pizza})}.
\end{remark*}

\subsection{From concrete to abstract syntax}\label{sec:CtoA}

From a syntactic point of view, N3 can be seen as a superset of Turtle~\cite{Prudhommeaux2014RT}, with which we assume the reader to be familiar.
The main N3-specific features of the concrete syntax were introduced in Sections~\ref{motivation}. 
N3 supports 
implicitly scoped universal variables, starting with a question mark~(\textit{e.g.} \texttt{?x});  graph terms, enclosed in curly brackets~(\texttt{\{\}});  and the symbol \texttt{=>} which is an alias for the predicate \texttt{log:implies}.
%
%
IRIs in the concrete syntax translate to IRIs~($C$) in the abstract syntax
(except when quantified with the keywords \texttt{@forAll} and \texttt{@forSome}, see below).
Literals in the concrete syntax translate to literals~($L$) in the abstract syntax.
Blank nodes and \texttt{?}-prefixed universals translate to variables~($V$).
Different concrete symbols are translated to distinct variables in the abstract syntax. In particular \texttt{?x} and \texttt{\_:x} are not the same and 
if we modify Formula~\ref{eq:love} as follows 
\begin{equation}\label{eq:love2}\tag{\ref*{eq:love}$'$}
	\texttt{?x :loves \_:x.}
\end{equation}
it corresponds to the abstract N3 graph
\begin{equation}\label{eq:abstractLove}
	\G{x_1}{x_2}{(x_1, \text{:loves}, x_2)}
\end{equation}

For the translation of implicit quantification, we adhere to the following schema: the variable corresponding to any \texttt{?}-prefixed universals is automatically added to the $U$ set of the outermost graph. The variable corresponding to a blank node is automatically added to the $E$ set of the innermost graph.
For example, in the RDF triple
\begin{equation}  \texttt{\_:x\ :likes\ :pasta.}\end{equation}
we have only one level, and the blank node is quantified there, we get
\begin{equation}  \G{}{x}{(x, \text{:likes}, \text{:pasta})}\end{equation}
while in Formula~\ref{eq:bob2} from above, the blank nodes occurring in the graph term are quantified there, and we get
\begin{equation}\label{eq:bob2-as}
		  \G{}{}{(\text{:bob}, \text{:says}, 	  \GT{}{x}{(x, \text{:likes}, \text{:pizza}),(x, \text{:likes}, \text{:pasta})})}
\end{equation}

In the following example:
\begin{equation}\label{eq:hasMother}
	  \texttt{\{?x\ a\ :Human\}\ =\textgreater{}\ \{?x\ :hasMother\ \_:y\}.\ }
\end{equation}
variable $\texttt{?x}$ is universally quantified on the outermost graph, the blank node $\texttt{\_:y}$ is  local to the consequent of the rule. We get
\begin{equation}\label{eq:humanMother}
	\G{x}{}{(
		\GT{}{}{(x, \text{rdf:type}, \text{:Human})},
		\text{log:implies},
		\GT{}{y}{(x, \text{:hasMother}, y)}  
		)}
\end{equation}
which is an existential rule~\cite{exN3}.

Explicit quantification in the concrete syntax is expressed with the \texttt{@forAll} or \texttt{@forSome} keyword, followed by an IRI and by a period~(\texttt{.}).
Such quantified IRIs will then translate to variables~($V$) in the abstract syntax rather than IRIs $(C)$.
Formula~\ref{eq:love} from above could also be written as
\begin{equation}\label{eq:love3}\tag{\ref*{eq:love}$''$}
	\texttt{@forAll :x. @forSome :y. :x :loves :y.}
\end{equation}
The formula again corresponds to the abstract formula~\ref{eq:abstractLove}.
The quantifier keywords can be used in the top-level graph or in any graph term, but must always occur before the first triple of the graph to which they belong, and all \texttt{@forAll}s must be placed before any \texttt{@forSome} of the same graph\footnote{
  These constraints were not part of the original N3 syntax, but we add them to improve readability and avoid corner cases.
}. Any further use of the quantified IRI in the same graph, or nested graph terms, will be represented in the abstract syntax by the corresponding variable. For example, the formula 
\begin{equation}\label{eq:expli}
	 \texttt{@forSome\ :y.\ :y\ :says\ \{ @forAll\ :x.\ :y\ :likes :x\}.\ }
\end{equation}
is represented as 
\begin{equation}
	  \G{}{y}{(y, \text{:says}, \GT{x}{}{(y, \text{:likes}, x})}
\end{equation}

\hypertarget{variables}{%
\subsection{Variables}\label{variables}}

For the definitions which follow later in this document, we need to make
some statements about the variables occurring in N3 graphs. Here we
distinguish between variables occurring directly in a graph an those
which are nested. To be able to differentiate them, we introduce the
following concept of constituents.

\begin{definition}[Direct Constituent]
  A term is a \emph{direct constituent} of a graph $G=(U, E, F)$ if it is the subject, predicate or object of a triple $t \in F$.
%
We denote the set of direct constituents of a graph \(G\) as
  \(DC(G)\).
  \end{definition}
We now define the more general notion of constituents, which are informally speaking all terms occurring in a graph.
  
\begin{definition}[Constituent]
  A term is a \emph{constituent} of a {graph} \(G_1\) if it is either
  \begin{itemize}
  \item
    a {direct constituent} of \(G_1\), or
  \item
    a {constituent} of graph \(G_2\), where the {graph term}
    \GTW{G_2}
    is a {direct constituent} of \(G_1\).
  \end{itemize}
We denote the set of constituents of a graph \(G\) as \(C(G)\).

\end{definition}
In other words: direct constituents are all terms which directly form
the graph\textquotesingle s triples. 
Constituents also include all terms found in nested graph terms independently of the depth of the nesting.


We illustrate these definitions on an earlier example: Formula~\ref{eq:bob2}, whose abstract syntax is given in Formula~\ref{eq:bob2-as}. 
%
   We encounter three  \emph{direct constituents}:
    \text{:bob}, \text{:says}, $\GT{}{x}{(x, \text{:likes}, \text{:pizza}), (x, \text{:likes}, \text{:pasta})}$;
and seven 
    \emph{constituents}:
    all the above plus
    \(x,\text{:likes}, \text{:pizza}, \text{:pasta}\).

With that definition, we now introduce the notion of free variable.
\begin{definition}[Free Variable]
\label{def:free-variables}
The set of \emph{free variables} in a term,  triple, set of triples or graph  $x$ is
noted $FV(x)$ and defined as follows:
\begin{itemize}
\item
$FV(x) = \emptyset$ if $x \in C \cup L$,
\item
$FV(x) = \{x\}$ if $x \in V$,
\item
$FV(x) = FV(s) \cup FV(p) \cup FV(o)$ if $x=(s,p,o)$ is a {triple},
\item
$FV(x) = \cup_{t \in x} FV(t)$ if $x$ is a set of {triples},
\item
$FV(x) = FV(F) \setminus(U \cup E)$ if $x=(U, E, F)$ is a
{graph},
\item
$FV(x) = FV(G)$ if $x=\GTW{G}$ is a
{graph term}.
\end{itemize}
For a graph \(G = (U, E, F)\) and a variable \(v \in C(G) \cap V\):
 We say that $v$ is \emph{free} in $G$ iff  $v\in FV(G)$.
Otherwise, $v$ is said to be \emph{scoped} in $G$.
\end{definition}
Examples:
\begin{itemize}
\item
  In \G{x}{y}{(x, \text{:loves}, y)}~(Formula~\ref{eq:love}) all variables are
  scoped.
\item
  In \G{}{y}{(x, \text{:knows}, y)} the variable \(x\) is
  free while \(y\) is scoped.
\item
  In
  \G{}{}{(y, \text{:says}, \GT{}{x}{(x,\text{rdf:type}, \text{:Mortal})})},
  the variable \(x\) is scoped while the variable \(y\) is free.
\end{itemize}
%
With this definition, we now introduce \emph{ground} graphs:
\begin{definition}[Ground term, triple and graph]
~
\begin{itemize}
\item A term is \emph{ground} if it is either
\begin{itemize}
\item an IRI,
\item a Literal,
\item a {graph term} \GTW{G} with $FV(G)=\emptyset$.
\end{itemize}
We denote the set of \emph{ground terms} as \(T_G\).

\item A triple $(s, p, o)$ is \emph{ground} if $s$, $p$ and $o$ are \emph{ground}.

\item A graph \(G=(U,E,F)\) is \emph{ground} if \(E=U=\emptyset\) and all {triples}
in \(F\) are {ground}.
\end{itemize}
\end{definition}

Note that this definition treats graphs and graph terms differently: a graph term $\GTW{(U, E, F)}$ may be ground while the corresponding graph $(U, E, F)$ is not.
More precisely: a graph term is ground if all its variable constituents are scoped, while to be ground, a graph must additionally have no variable in its direct constituents.
A ground graph may still contain variables,
as long as these are scoped within a graph term.  This is because  these variables are not accessible by outside quantification.
If we modify Formula~\ref{eq:bob2-as} from above by adding quantification we get
\begin{equation}\label{eq:bob2-as3}
	\G{x}{}{(\text{:bob}, \text{:says}, 	  \GT{}{x}{(x, \text{:likes}, \text{:pizza}),(x, \text{:likes}, \text{:pasta})})}.
\end{equation}
As the variable $x$ in this graph only occurs, scoped, in the graph term $\GT{}{x}{(x, \text{:likes}, \text{:pizza}),(x, \text{:likes}, \text{:pasta})})$, it can not be "reached" by the universal quantifier we just added, and therefore 
Formulae~\ref{eq:bob2-as} and~\ref{eq:bob2-as3} will be interpreted equally.


The example illustrates that a graph \(G=(U,E,F)\) can contain quantified variables
\(v\in U\cup E\) which do not occur freely in any triple of \(F\) and thus do not contribute to the interpretation of the graph. In order
to simplify our considerations for the following sections, we introduce
the notion of normalised graphs which exclude such quantification.

\begin{definition}[Normalisation]
\label{def:normalisation}
For each graph \(G=(U,E,F)\), we call
\(G^N=(U\cap (FV(F)), E\cap(FV(F)), F)\) the \emph{normalisation} of \(G\).
\end{definition}
The normalisation of Formula~\ref{eq:bob2-as3} is Formula~\ref{eq:bob2-as}. We provide more examples below:
\begin{itemize}

\item
  \(
  \G{x,z}{y}{(x, \text{:knows}, y)}^N
  =
  \G{x}{y}{(x, \text{:knows}, y)}
  \)
\item
  \(
  \G{x_1, x_2, x_3}{y_1, y_2, y_3}{(\text{:socrates}, \text{:knows}, \text{:plato})})^N
  =
  \G{}{}{(\text{:socrates}, \text{:knows}, \text{:plato})}
  \)
\item
  \(
  \G{x}{y}{(y, \text{:says}, \GT{}{x}{(x,\text{rdf:type}, \text{:Mortal} )})}^N
  \) \\
  \(
  =
  \G{}{y}{(y, \text{:says}, \GT{}{x}{(x,\text{rdf:type}, \text{:Mortal} )}}
  \)
\end{itemize}

\hypertarget{isomorphisms}{%
\subsection{Isomorphisms}\label{isomorphisms}}
Having the syntactic definition in place, we can now take a closer look at the relation between two graphs. Just like in RDF, we can define isomorphisms
 between N3 graphs and between graph terms. The latter will become particularly important in Section~\ref{semantics} where we define the semantics. When interpreting graph terms,  the naming of a variable should not matter. 
Just as the Formulae~\ref{eq:love},~\ref{eq:love2} and~\ref{eq:love3} had the same meaning and could be translated to the abstract formula~\ref{eq:abstractLove}, the latter formula should have the meaning as
\begin{equation}\label{eq:abstractLove2}\tag{\ref*{eq:abstractLove}$'$}
	\G{y}{z}{(y, \text{:loves},z)}
\end{equation}
where we only changed the variable names.  And that behavior should also hold for graph terms: we would like our semantics to assign the same meaning to the abstract graph term 
\begin{equation}\label{eq:trT}
	\GT{}{x}{(x, \text{:likes}, \text{:pizza}),(x, \text{:likes}, \text{:pasta})}
\end{equation}
from Formula~\ref{eq:bob2-as} as for 
\begin{equation}\label{eq:trT2}\tag{\ref*{eq:trT}$'$}
	\GT{}{y}{(y, \text{:likes}, \text{:pizza}),(y, \text{:likes}, \text{:pasta})}
\end{equation}
where we only replaced the variable $x$ by $y$.


In order to be able to define isomorphisms between graph terms, we first introduce mechanisms to combine mappings. Here and 
later in this paper,  we write $dom(M)$ to denote the \emph{domain} $X$ of a mapping $M:X\rightarrow Y$.

\begin{definition}[Mapping Combination]
Let \(M\) and \(N\) be two mappings from terms to terms, we define their
\emph{combination}, noted \(M\bullet N\) as follows:
\[
M\bullet N(v)=\begin{cases}N(v) \text{ if } v\in dom(N)\\
M(v) \text{ if } v\not\in dom(N) \text{ and } v\in dom(M)
  \end{cases}
\]
\end{definition}
Informally speaking, this definition means that the second mapping
``overrides'' the first one whenever a term is in both of their domains. This is very similar to the behavior of quantifiers in
first-order logic where the scope of a quantifier is always the set of
quantified variables occurring freely in the quantified formula. 
This similarity is no coincidence. In  our semantics, combinations will be mainly used to handle nested variable scopes.

For a mapping $M:X\rightarrow S$, we denote by
$M|_{Y}$ the restriction of \(M\) to the set $Y\subset X$.

We can  now define what it means for N3 graphs to be isomorphic to each other.

\begin{definition}[Isomorphic under $M$]
Let \(M\) be a bijection between sets of variables in $V$:
\begin{itemize}
\item
  a graph \(G_1\) is  \emph{isomorphic to \(G_2\) under \(M\)} iff,
  \(G_1^N =(U_1,E_1,F_1)\) and \(G_2^N=(U_2,E_2,F_2)\) and
  \begin{itemize}
  \item
    \(M(U_1) := \{M(u)|u\in U_1\}= U_2\) and
    \(M(E_1):= \{M(e)|e\in E_1\}= E_2\)
  \item
    \(F_1\) is \emph{isomorphic to \(F_2\) under \(M\)}
  \end{itemize}
\item
  a set of triples \(F_1\) is isomorphic  to a set of triples
  \(F_2\) under \(M\) iff 
%
\begin{itemize}
	\item $M$ is a bijection from $FV(F_1)$ to $FV(F_2)$,
	\item   there is a bijection $T$ 
	from the triples in \(F_1\) to those in \(F_2\) such
	that, for each triple \(t_1=(s_1, p_1, o_1)\in F_1\) and
	for the image \(T((s_1, p_1, o_1))=
	(s_2, p_2, o_2)\in F_2\) 
	it holds that 
	\(s_1\), \(p_1\) and
	\(o_1\) are isomorphic under \(M\) to \(s_2\), \(p_2\) and \(o_2\),
	respectively
\end{itemize}
\item
  a variable \(v\) is \emph{isomorphic under $M$} to \(M(v)\),
\item
  an IRI or Literal is \emph{isomorphic under $M$} to itself,
\item
  a graph term \GTW{G_1} is \emph{isomorphic under \(M\)} to \GTW{G_2} if
  for \(G_1^{N} =(U_1,E_1,F_1)\) and \(G_2^{N}=(U_2,E_2,F_2)\) there is
  a bijection \(P\) between \(U_1 \cup E_1\) and \(U_2 \cup E_2\) such
  that  \(M|_{FV(G_1)}\bullet P\) is a bijection and \(G_1\) is  \emph{isomorphic to \(G_2\) under $M|_{FV(G_1)}\bullet P$}.
\end{itemize}
\end{definition}

As we will further discuss below~(Lemma~\ref{lemma:eq}), if a graph $G_1$ is isomorphic to a graph $G_2$ under $M$, then $G_2$ is isomorphic to $G_1$ under $M^{-1}$. This close connection motivates us to omit mentioning the mapping when talking about isomorphisms. We get the following more general definition: 


\begin{definition}[Isomorphic]
A graph \(G_1\) is said to be \emph{isomorphic} to a graph \(G_2\), written as
\(G_1\simeq G_2\), if there exists some \(M\) such that \(G_1\) is
isomorphic to \(G_2\) under \(M\).
\end{definition}

We note that, for the subset of N3 graphs that are also RDF graphs, our notion of isomorphism is equivalent to that defined in the RDF specification~\cite{rdf11-concepts}.
We also note the following lemma.

\begin{lemma}[Equivalence property of Isomorphisms]\label{lemma:eq}
	The isomorphic relation from above  is reflexive, symmetric, and transitive.
\end{lemma}

To prove this lemma, we first remark that, given a graph $G$ such that $G^N = (U, E, F)$, it holds that $FV(F) = U \cup E \cup FV(G)$ per Definitions~\ref{def:free-variables} and~\ref{def:normalisation}. For convenience, we denote this set by $var(G)$.

Given two mappings $A$ and $B$ with disjoint domains, we remark that $A \bullet B = A \cup B$~(treating mappings as set of pairs). If additionally $A$, $B$ and $A \bullet B$ are bijections, this implies that $(A \bullet B)^{-1} = A^{-1} \bullet B^{-1}$.

\begin{proof}[Isomorphism of N3 graphs~($\simeq$) is reflexive]
\label{reflexivity}
This is trivially true by choosing the identity for the bijections $M$, $T$ and $P$.
\end{proof}

\begin{proof}[Isomorphism of N3 graphs~($\simeq$) is symmetric]
\label{symmetry}
Let us first prove by induction that if \(G_1\) is isomorphic to \(G_2\)
under a bijection \(M\), then \(G_2\) is isomorphic to \(G_1\) under
\(M^{-1}\). Let $G_1 = (U_1, E_1, F_1)$ and
$G_2 = (U_2, E_2, F_2)$. There must be a bijection \(T\) mapping triples from
\(F_1\) to triples from \(F_2\) that satisfies the definition of {isomorphism} under
\(M\).

If we assume that neither \(G_1\) nor \(G_2\) contain graph terms, then
trivially \(G_2\) is isomorphic to \(G_1\) under \(M^{-1}\) by choosing
\(T^{-1}\) to map from triples from \(F_2\) to \(F_1\).

If \(G_1\) contains graph terms, then \(T\) must map every graph term
\GTW{H_1} occurring in \(G_1\) to a graph term \GTW{H_2} occurring
in \(G_2\), and for each such pair there must exist a bijection \(P\)
such that \(H_1\) is isomorphic to \(H_2\) under
\(M|_{FV(H_1)} \bullet P\). 

Per our induction hypothesis, $H_2$ is
also isomorphic to $H_1$ under $(M|_{FV(G1)} \bullet P)^{-1}$.
We remark that the domains of $M|_{FV(H_1)}$ and $P$ are disjoint, 
so $(M|_{FV(H_1)} \bullet P)^{-1} = (M|_{FV(H_1)})^{-1} \bullet P^{-1} = M^{-1}|_{FV(H_2)} \bullet P^{-1}$. We can then conclude that the graph term \GTW{H_2} is isomorphic to \GTW{H_1} under $M^{-1}$, choosing $P^{-1}$, and therefore that $G_2$ is isomorphic to $G_1$ under $M^{-1}$ 

From the above, it follows that if \(G_1 \simeq G_2\), then
\(G_2 \simeq G_1\).
\end{proof}

\begin{proof}[Isomorphism of N3 graphs~($\simeq$) is transitive]
\label{transitivity}

Let us first prove by induction that if \(G_1\) is isomorphic to \(G_2\)
under \(M_1\), and \(G_2\) is isomorphic to \(G_3\) under \(M_2\), then
\(G_1\) is isomorphic to \(G_3\) under \(M_2 \circ M_1\). Let
\(G_i^N = (U_i, E_i, F_i)\) for \(i \in \{1, 2, 3\}\). There must be a
bijection \(T_1\) from \(F_1\) to \(F_2\) and a bijection \(T_2\) from
\(F_2\) to \(F_3\) that satisfy the definition of {isomorphism} under
\(M_1\) and \(M_2\) respectively.

If we assume that neither \(G_1\), \(G_2\) nor \(G_3\) contain graph
terms, then trivially \(G_1\) is isomorphic to \(G_3\) under
\(M_2 \circ M_1\) by chosing \(T_2 \circ T_1\) to map from \(F_1\) to
\(F_3\).

If \(G_1\) contains graph terms, then \(T_1\) must map every graph term
\GTW{H_1} occurring in \(G_1\) to a graph term \GTW{H_2} occurring
in \(G_2\), and for each such pair there must exist a bijection \(P_1\)
such that \(H_1\) is isomorphic to \(H_2\) under
\(M_1|_{FV(H_1)} \bullet P_1\). Similarly \(T_2\) must map every graph
term \GTW{H_2} occurring in \(G_2\) to a graph term \GTW{H3}
occurring in \(G_3\), and for each such pair there must exist a
bijection \(P_2\) such that \(H_2\) is isomorphic to \(H_3\) under
\(M_2|_{FV(H_2)} \bullet P_2\). Per our induction hypothesis, \(H_1\) is
isomorphic to \(H_3\) under
$(M_2|_{FV(H_2)} \bullet P_2)\circ(M_1|_{FV(H_1)} \bullet P_1)$.
Since the domain of $M|_{FV(H_i)}$ and $P_i$ are disjoint for all $i$,
this is equal to
$(M_2|_{FV(H_2)} \cup P_2)\circ(M_1|_{FV(H_1)} \cup P_1)$.
Since the image of $M_1|_{FV(H_1)}$ is $FV(H_2)$, and the image of $P_1$ is the domain of $P_2$, this is equal to
$(M_2|_{FV(H_2)} \circ M_1|_{FV(H_1}) \cup (P_2 \circ P_1)
= (M_2 \circ M_1)|_{FV(H_1)} \bullet (P_2 \circ P_1)$
and therefore each graph term \GTW{H_1} is isomorphic to \GTW{H_3} under \(M_2 \circ M_1\), choosing $P_2 \circ P_1$.
Hence, \(G_1\) is isomorphic to \(G_3\) under \(M_2 \circ M_1\).

From the above, it follows that if \(G_1 \simeq G_2\) and
\(G_2 \simeq G3\), then \(G_1 \simeq G_3\).
\end{proof}

We give an example to further illustrate the definition. Consider the following graphs  in concrete syntax
\begin{align*}
	G_1= & \{\texttt{\_:x\ :knows\ :p.\ :p\ :says\ \{\_:x\ a\ :Genius\}.}\}\\
	G_2=& \{\texttt{\_:y\ :knows\ :p.\ :p\ :says\ \{\_:x\ a\ :Genius\}.}\}\\
	G_3=&  \{\texttt{\_:y\ :knows\ :p.\ :p\ :says\ \{\_:x\ a\ \_:z\}.}\}
\end{align*}
written in abstract syntax these are 
\begin{align*}
	G_1=&  \G{}{x}{
		(x, \text{:knows}, \text{:p}),
		(\text{:p}, \text{:says}, \GT{}{x}{(x, \text{rdf:type}, \text{:Genius})})
	}\\
	G_2=&   \G{}{y}{
		(y, \text{:knows}, \text{:p}),
		(\text{:p}, \text{:says}, \GT{}{x}{(x, \text{rdf:type}, \text{:Genius})})
	}\\
	G_3=&   \G{}{y}{
		(y, \text{:knows}, \text{:p}),
		(\text{:p}, \text{:says}, \GT{}{x,z}{(x, \text{rdf:type}, z)})
	}
\end{align*}
Then \(G_1\) is isomorphic to \(G_2\) but not to \(G_3\).

\hypertarget{other_graph_relationships}{%
\subsection{Syntactic relations between graphs }\label{other_graph_relationships}}

Beyond isomorphism, N3 graphs can have other useful relationships with each other.
In this section we define the notion of subgraph, instance/generalisation, and merge.

\begin{definition}[Subgraph]
	A graph $G_1=(U_1, E_1, F_1)$ is a \emph{subgraph} of a graph  $G_2=(U_2, E_2, F_2)$, written as $G_1\subseteq G_2$, if $FV(G_1)\subseteq FV(G_2)$, $U_1\subseteq U_2$, $E_1\subseteq E_2$, and $F_1\subseteq F_2$.
\end{definition}
Note that this definition is slightly more complicated than the subgraph relation in RDF which is basically the same as a subset~\cite{RDFsemantics}, since in N3, we need to also take variable bindings into account. Given the  graphs
\[G=  \G{x}{y}{
	(x, \text{:knows}, y),
	(x, \text{rdf:type}, \text{:Logician})}\]
and 
\begin{align*}
	G_4=&  \G{x}{}{
		(x, \text{rdf:type}, \text{:Logician)}}\\
	G_5=&  \G{}{x}{
		(x, \text{rdf:type}, \text{:Logician)}}\\
	G_6=&  \G{x}{}{
		(x, \text{:knows}, y),
		(x, \text{rdf:type}, \text{:Logician)}}
\end{align*}
$G_4$ is a subgraph of $G$, but neither are $G_5$~(quantifiers differ) nor $G_6$~($y$ is free in $G_6$).

Note that 
for the definition of subgraphs, concrete variable names matter.
The graph
\[G_7=  \G{z}{}{
	(z, \text{rdf:type}, \text{:Logician})}\]
is not a subgraph of $G$, it is, however, isomorphic to one of its subgraphs, namely $G_4$.

\vspace{1ex}
In order to define the notions of instance and generalisation, we first need to define a way to properly substitute variables in an N3 graph.

\begin{definition}[Total application]\label{def:total-application}
Let  \(M\) be a mapping from  a set of variable \(W\subset V\) into a set~$S$. 
A \emph{total application} \(M^t\) of \(M\) on graphs and terms is defined as follows:
\TabPositions{4.5cm}
\begin{itemize}
\item
  if $x$ is a variable \tab
  $M^t(x) = M(x)$ if $x\in dom(M)$, $x$ otherwise
\item
  if $x$ is an IRI or a Literal \tab
  $M^t(x)=x$
\item
  if $x=\GTW{H}$ is a graph term \tab
  $M^t(x)=\GTW{M^t(H)}$
\item
  if $x=(s, p, o)$ \tab
  $M^t(x)=(M^t(s), M^t(p) ,M^t(o))$,
\item
  if $x = F$ is a set of triples \tab
  $M^t(x)=\{ M^t(t) | t\in F\}$
\item
  if $x=(U, E, F)$ is a graph \tab
  $M^t(x)=(U, E, M|_{FV(x)}^t(F))$
\end{itemize}
%
\end{definition}

The total application of $M$ replaces all free variables occurring directly or
nested in a graph with their image through $M$. As an example we consider a modification of the
rule in Formula~\ref{eq:hasMother}, where the universal \texttt{?x} is replaced with a free variable:
\begin{equation}
	\G{}{}{(
		\GT{}{}{(x, \text{rdf:type}, \text{:Human})},
		\text{log:implies},
		\GT{}{y}{(x, \text{:hasMother}, y)}  
		)}
\end{equation}
Let \(M:\{x,y\}\rightarrow S\)
be a mapping such that \(M(x)=tom\) and \(M(y)=mary\), then we get:
\begin{multline}
M^t(\G{}{}{(
\GT{}{}{(x, \text{rdf:type}, \text{:Human})},
\text{log:implies},
\GT{}{y}{(x, \text{:hasMother}, y)\})}
)})\\
=
\G{}{}{(
\GT{}{}{(tom, \text{rdf:type}, \text{:Human})},
\text{log:implies},
\GT{}{y}{(tom, \text{:hasMother}, y)\})}
)}
\end{multline}
In this example, we replace $x$ with its value because it occurs freely in
the graph while $y$ does not get replaced since it is quantified in the
nested graph.


\begin{definition}[Instance and generalisation of a graph]
\label{def:instance}
Let $G=(U, E, F)$ be a graph.

We call a graph $G'$ an \emph{E-instance} of $G$ if there exists a mapping 
$M:X\subseteq E\rightarrow T_G \cup E$ such that $G'$  is the normalisation of $(U, E, M^t(F))$.
Conversely, we call $G$ an \emph{E-generalisation} of $G'$.

We call a graph $G''$ a \emph{U-instance} of $G$ if there exists a mapping 
$M:X\subseteq U\rightarrow T_G \cup U $ such that 
$G''$  is the normalisation of $(U, E, M^t(F))$.
Conversely, we call $G$ a \emph{U-generalization} of $G''$.
\end{definition}

For example:
\begin{itemize}
\item
$\G{}{}{(\text{:alice}, \text{:likes}, \text{:pizza})}$ is an E-instance of $\G{}{x}{(x, \text{:likes}, \text{:pizza})}$, with $M = \{{x\mapsto \text{:alice}}\}$
\item
$\G{}{x}{(x, \text{:likes}, \text{:pizza}), (x, \text{:likes}, \text{:pasta})}$ is an E-instance of $\G{}{x, y}{(x, \text{:likes}, \text{:pizza}), (y,\allowbreak \text{:likes}, \text{:pasta})}$, with $M = \{y \mapsto x\}$
\item
$\G{}{}{(\text{:alice}, \text{:likes}, \text{:pizza})}$ is a U-instance of $\G{x}{}{(x, \text{:likes}, \text{:pizza})}$, with $M = \{{x\mapsto \text{:alice}}\}$
\item
 $\G{x}{}{(x, \text{:likes}, x)}$ is a U-instance of $\G{x, y}{}{(x, \text{:likes}, y)}$ with $M = \{y \mapsto x\}$.
\item
$\G{}{x}{(x, \text{:likes}, \text{:alice})}$ is an E-generalization of $\G{}{}{(\text{:bob}, \text{:likes}, \text{:alice})}$, which is itself a U-instance of $\G{y}{}{(y, \text{:likes}, \text{:alice})}$,
\end{itemize}

We will see below~(Section~\ref{sec:interpolation-lemma}) that these definitions allow us to formulate a counterpart to RDF's interpolation lemma~\cite[Section 5.2]{RDFsemantics}. In particular, note that our notion of E-instance extends the notion of \emph{instance} of an RDF graph. In contrast to RDF, N3 also covers universal quantification, hence the need to differentiate two kinds of instances/generalisations.

\begin{definition}[Merge of graphs]\label{def:merge}
Let $G_1=(U_1,E_1,F_1)$ and $G_2=(U_2,E_2,F_2)$ be two graphs such that $U_1 \cap U_2 = E_1 \cap E_2 = FV(G_1) \cap FV(G_2) = \emptyset$. By $G_1\sqcup G_2:=(U_1 \cup U_2,E1 \cup E2,F_1 \cup F_2)$ we denote the disjoint union of $G_1$ and $G_2$.

Let $\{G_1, \dots, G_n\}$ be a set of graphs. We call a graph $G$ a \emph{merge} of $G_1, \dots, G_n$ if, for all $i \in [1, n]$ there exist a graph $G'_i$ isomorphic to $G_i$, such that $G = G'_1 \sqcup \dots \sqcup G'_n$. 
\end{definition}

This definition extends that of the \emph{merge} of RDF graphs as defined in Section~4.1 of RDF Semantics~\cite{RDFsemantics}.

\begin{definition}[Closure of a graph]
A graph $G$ is said to be \emph{closed} if it does not contain any free variable, i.e. $FV(G)=\emptyset$.
Given a graph $G = (U, E, F)$, the closed graph $(U \cup FV(G), E, F)$ is called the \emph{closure} of $G$.
\end{definition}

As we will see in the following section, any graph is semantically equivalent to its closure.

\section{Semantics}\label{semantics}
Having introduced N3's abstract syntax  in the previous section, we now define N3's semantics. We do that in two steps: we first define N3's base semantics which only covers the meaning of terms and variables but does not take any special predicates into account. After that, we discuss N3 semantics which assigns meaning to N3's built-in predicates~(Section~\ref{n3-entailment}). Here, we limit our considerations to the most important predicates~(including \texttt{log:implies}). We plan to extend the definitions in future work.
\hypertarget{base-semantics}{%
\subsection{Base Semantics}\label{base-semantics}}
To define N3's base semantics, we start with an interpretation function:
\begin{definition}[Interpretation]\label{def:interpretation}
An {interpretation} \(I\) is defined by:
\begin{itemize}
\item
  a set \(\Delta_I\)~({domain}) whose elements are referred to as \emph{resources},
\item
  a mapping \(D_I\)~({denotation}) from \(C ∪ L\) to \(\Delta_I\), such that
  its restriction to \(C\) is a total mapping,
\item
  a mapping \(Q_I\) from {ground} {graph terms} to ground graph terms,
  such that
  \begin{itemize}
  \item
    for any graph term \GTW{G}: \(Q_I(\GTW{G})=\GTW{H} \implies H \simeq G\)
  \item
    for two graph terms \(\GTW{G_1}\) and \(\GTW{G_2}\):\\
    \(G_1 \simeq G_2\implies Q_I(\GTW{G_1})=Q_I(\GTW{G_2})\)
  \end{itemize}
\item
  a subset \(EXT_I\)~({extension}) of \(\Delta_I^3\).
\end{itemize}
\end{definition}
This definition is generally ``compatible'' with RDF simple
interpretations, with the following superficial changes:
\begin{itemize}
\item
  IR becomes \(\Delta_I\),
\item
  IP is assumed, without loss of generality, to be a subset of \(\Delta_I\),
\item
  IS and IL have been merged in \(D_I\),
\item
  IEXT can be derived from \(EXT_I\); formally:
  \(IEXT(p) = \{(s,o)|(s,p,o) ∈ EXT_I\}\); conversely:
  \(EXT_I = \bigcup_{p∈IP} \{(s,p,o)| (s,o) ∈ IEXT(p)\}\).
\end{itemize}
The interpretation maps each ground graph term to an isomorphic copy of itself. In that sense, ground graph terms are similar to literals in RDF~(they have a ``fixed'' denotation). As a consequence, two non-isomorphic graphs cannot have the same meaning.

%
%
%
However,  N3-entailment(see Section~\ref{n3-entailment}) provides mechanisms to ``unquote''
graphs and to compare them based on the meaning of the formula they quote.


We now inductively define how to interpret {ground} graphs with an interpretation $I$:
\begin{definition}[Application of an interpretation]
For ground terms, triples, and graphs, we define the application of an interpretation $I$ as follows:
\TabPositions{5.5cm}
\begin{itemize}
\item
  if \(x\) is an IRI or a Literal \tab
  \(I(x)=D_I(x)\)
\item
  if $x$ is a ground graph term \tab
  \(I(x)=Q_I(x)\)
\item
  if $x=(s, p, o)$ is a ground triple \tab
  \(I(x)=true\) if \((I(s), I(p) ,I(o))\in EXT_I\),

  \tab $false$ otherwise
\item
  if $x = F$ is a set of ground triples \tab
  $I(x)=false$ if $I(t)=false$ for some $t \in F$,

  \tab $true$ otherwise
\item
  if $x=(\emptyset, \emptyset, F)$ is a ground graph \tab
  $I(x)=I(F)$
\end{itemize}
\end{definition}

The definition implies that the empty graph is true under any interpretation $I$.


In order to extend the definition towards non-ground graphs,
we rely on a few additional definitions. To better understand the need
for these, consider the following example:
\begin{equation}
\texttt{?x\ :says \ \{?x\ :is\ :wise\}.}
\end{equation}
which freely translates to \emph{``everyone says about themselves that they are wise''}.
In abstract syntax, the formula would be expressed as:
\begin{equation}\label{eq:universalWisdom}
\G{x}{}{(x, \text{:says}, \GT{}{}{(x, \text{:is}, \text{:wise)})}}
\end{equation}
This formula is interesting because the universally quantified variable appears once as a {direct constituent} in the triple and once nested.
The semantics for ground graphs maps direct constituents to~(not further specified) domain elements, while ground nested graphs are mapped to an isomorphic version of themselves.

If we replace the variable \(x\) with,
for example, \texttt{:bob}, we get
\begin{equation}\label{eq:wiseBob}
\G{}{}{(\text{:bob}, \text{:says}, \GT{}{}{(\text{:bob}, \text{:is}, \text{:wise)})}}
\end{equation}
With an interpretation $I$: \\
$\begin{aligned}
&   I(\G{}{}{(\text{:bob}, \text{:says}, \GT{}{}{(\text{:bob}, \text{:is}, \text{:wise)})})}) = \texttt{true}
 \quad \text{(interpretation of the graph)}\\
  \text{if}\quad
 &  I(~~(\text{:bob}, \text{:says}, \GT{}{}{(\text{:bob}, \text{:is}, \text{:wise)})})~)
  = \texttt{true} \quad\text{(interpretation of the triple), which holds}\\
  \text{if}\quad
  &  (D_I(\text{:bob}), D_I(\text{:says}), Q_I(\GT{}{}{(\text{:bob}, \text{:is}, \text{:wise)})}))
    \in  EXT_I 
\end{aligned}$

Note that the mapping \(D_I\) maps the URI \texttt{:bob} directly to a
resource of \(Δ_I\), while \(Q_I\) operates on the graph term which
contains the URI \texttt{:bob}. If we go back to Formula~\ref{eq:universalWisdom}
where we have a variable $x$ instead of \texttt{:bob},
the first occurrence of $x$ should stand for a resource,
while the second occurrence~(in the graph term) depends on the syntactical representation of that resource.
Note that this diverges from how RDF deals with variables, i.e.
blank nodes~\cite[Section 5.1]{RDFsemantics}.

\vspace{1em}
While RDF semantics maps blank nodes to resources,
our semantics for N3 maps variables to ground terms,
which are then, depending on the context,
mapped to resource via an interpretation $I$,
or kept as is.
While this shift is not neutral
(the domain may contain unnamed elements, which are therefore unreachable by quantified variables),
we will prove later that it has no significant impact on basic entailment:
N3 basic entailment on RDF graphs is equivalent to RDF simple entailment,
despite this different treatment of existentials~(Lemma~\ref{lemma:simpleEntailment} below).
The difference is likely to matter, however, with more complex inference regimes introducing negation~\cite{Genesereth2015}.

\begin{definition}[Interpretation with assignment]
Let \(I\) be an interpretation and
 $A:V\rightarrow T_G$  a mapping~(assignment) from $V$ to $T_G$.
We define the application \(I[A](x)\)  on terms, triples, and graphs as follows:
\TabPositions{4.5cm}
\begin{itemize}
\item
  if $x$ is a variable \tab
  \(I[A](x)=I(A(x))\)
\item
  if \(x\) is an IRI or a Literal \tab
  \(I[A](x)=D_I(x)\)
\item
  if $x$ is a graph term \tab
  \(I[A](x)=Q_I(A^t(x))\)
\item
  if $x=(s, p, o)$ \tab
  \(I[A](x)=true\) if \((I[A](s), I[A](p) ,I[A](o))\in EXT_I\),

  \tab $false$ otherwise
\item
  if $x = F$ is a set of triples \tab
  $I[A](x)=false$ if $I[A](t)=false$ for some $t \in F$,

  \tab $true$ otherwise
\item
  if $x=(U, E, F)$ is a graph \tab
  $I[A](x)=true$ if for each mapping $B:U\rightarrow T_G$

  \tab there exists a mapping  $C:E\rightarrow T_G$

  \tab such that  $I[A\bullet B\bullet C](F)=true$,

  \tab $false$ otherwise
\end{itemize}
\end{definition}

When $x$ is a graph, $A$ will only affect the free variables in $x$
(quantified variables will be handled by $B$ and $C$).
Then if $x$ is a closed graph~(i.e. without free variables), $I[A](x)$ is not influenced by $A$.
In such cases, we can omit the assignment $A$ and simply write $I(x)$.

\begin{definition}[Basic models]\label{def:BM}
Let  $I$ be an interpretation and  $G$ a graph.
We call $I$ a \emph{basic model} of $G$, noted $I\models_b G$,
if $I[A](G)=true$ for all mappings $A:V\rightarrow T_G$.
\end{definition}

The above definition handles all variables occurring freely in a graph
as universally quantified.
Note that the interpretation of a variable \(x\) depends on its position in a graph.
In the expression \((x, \text{:p}, \GT{}{}{(x, \text{:q}, \text{:r})})\),
the first \(x\) is interpreted as \(I(A(x))\)~(an element of the domain)
while the second is interpreted as \(A(x)\)~(a ground term).



\begin{definition}[Basic entailment]\label{def:basicEntailment}
Let $G$ and $H$ be graphs. We say that $G$ b-entails $H$, written as $G\models_b H$ iff for each basic model $I$ of $G$, i.e. such that $I\models_b G$, it holds that $I\models_b H$.
\end{definition}

We give some examples for b-enatilemnt. Formula~\ref{eq:universalWisdom} from above b-entails Formula~\ref{eq:wiseBob} and
\begin{equation} \label{eq:ent3}		  \G{}{x}{
		(x, \text{:says}, \GT{}{}{(x, \text{is}, \text{:wise})})
	}
\end{equation}
but not
\begin{equation}
\label{eq:ent2a}		  \G{x}{}{
		(\text{x}, \text{:says}, \GT{}{}{(\text{:bob}, \text{is}, \text{:wise})})
	}
\end{equation}

Formula~\ref{eq:ent2a} b-entails Formulae~\ref{eq:wiseBob} and~\ref{eq:ent3}, and Formula~\ref{eq:wiseBob} b-entails Formula~\ref{eq:ent3}. 
Note that basic entailment does not support rules. Rule~\ref{eq:humanMother} from above

\begin{equation}\tag{\ref{eq:humanMother}}
	\G{x}{}{(
		\GT{}{}{(x, \text{rdf:type}, \text{:Human})},
		\text{log:implies},
		\GT{}{y}{(x, \text{:hasMother}, y)}  
		)}
\end{equation}
b-entails 
\begin{equation}\label{eq:humanMotherInst}
	\G{}{}{(
		\GT{}{}{(\text{:tom}, \text{rdf:type}, \text{:Human})},
		\text{log:implies},
		\GT{}{y}{(\text{:tom}, \text{:hasMother}, y)}  
		)}
\end{equation}
but even in combination with
\begin{equation}\label{eq:tom}
\G{}{}{(\text{:tom}, \text{rdf:type}, \text{:Human})}
\end{equation}
we cannot b-entail that 
\begin{equation}\label{eq:tomsMother}
\G{}{y}{(\text{:tom}, \text{:hasMother}, y)}.  
\end{equation}
This kind of entailment is only supported by N3-entailment which  will be discussed below in Section~\ref{n3-entailment}.

As basic entailment on its own does not support concepts like rules or negation, it can not be used to construct contradictions and all N3 graphs have a basic model.

\begin{lemma}[Basic satisfiability]
For each N3 graph there exists a basic model.
\end{lemma}
\begin{proof}
Consider the interpretation with $\Delta_I=T_G$, $D_I=\text{id}$ the identity mapping, $EXT_I=\Delta_I^3$ and $Q_I$ mapping all ground graph terms to one of their isomorphic representations~(such as required in Definition~\ref{def:interpretation}). This interpretation is a model for all graphs.
\end{proof}

For the closure of a graph we furthermore observe the following:
\begin{lemma}[Closure]
Each N3 graph has the same models as its closure.
\end{lemma}
\begin{proof}
Let $G_1 = (U, E, F)$ be an N3 graph and $G_2 = (U\cup FV(G_1), E, F)$ its closure.
Let $I \models_b G_1$. We show that for all $B_2: U\cup FV(G_1) → T_G$ there exists $C_2: E → T_G$ such that $I[B_2\bullet C_2](F) = \text{true}$. As $I \models_b G_1$, we know that for $A_1: V → T_G$ and $B_1: U → T_G$, there exists $C_1: E → T_G$. If we chose $A_1$ and $B_1$ such that $B_2 = A_1|_{FV(G_1)} \bullet B_1$, we set $C_2 := C_1$, then we have $I[B_2\bullet C_2](F) = \text{true}$ and therefore $I \models_b G_2$.

Conversely, let $I \models_b G_2$. We show that for all $A_1: V → T_G$ and $B_1: U → T_G$ there exists $C_1: E → T_G$ such that $I[A_1\bullet B_1\bullet C_1](F) = \text{true}$. As $I \models_b G_2$, we know that for $B_2 := A_1|_{FV(G_1)} \bullet B_1$, there exists $C_2: E → T_G$. We set $C_1 := C_2$, then we have $I[A_1\bullet B_1\bullet C_1](F) = \text{true}$ and therefore $I \models_b G_1$.
\end{proof}

From the lemma above, we conclude that we can, without loss of generality, focus on closed graphs for the rest of the proofs. Any subsequent result is straightforwardly extensible to non-closed graph by first taking their closure.

\vspace{1em}
For pure RDF graphs~(i.e. graphs without graph terms and universal quantification) b-entailmant coincides with RDF simple entailment~\cite[Section 5.2]{RDFsemantics}. For an RDF graph $G$ let $E$ be the set of blank nodes occurring in it.
We denote its N3 representation with explicit quantifiers as detailed in Section~\ref{sec:CtoA} as $G^{N3}:=(\emptyset, E, G)$. 
\begin{lemma}[Simple Entailment]\label{lemma:simpleEntailment}
For two RDF graphs $G_1$ and $G_2$ with blank nodes $E_1$ and $E_2$, respectively, the following holds:
\[G_1 \text{ simply entails } G_2 \quad \text{ iff } \quad G_1^{N3}=(\emptyset, E_1, G_1) \text{ b-entails } G_2^{N3}=(\emptyset, E_2, G_2) \]
\end{lemma}


\begin{proof}
In the following, we note RDF simple entailment as $\models_s$,
to clearly differentiate it from Notation3 basic entailment, still noted $\models_b$.
We also use RDF simple interpretations and Notation3 interpretations  interchangeably, since we showed after Definition~\ref{def:interpretation} that they can easily be transformed into each other.

$\Rightarrow$
Assume $G_1 \models_s G_2$.
Let $I \models_b G_1^{N3}$, i.e. there exists $B_1: E_1 → T_G$ such that $I[B_1](G_1^{N3})=\text{true}$.
We build a new model $I'$ such that $\Delta_I' = Img(D_I)$, and $EXT_{I'}$ is the restriction of $EXT_I$ to this new domain.
Clearly $I'[B_1](G_1^{N3})=\text{true}$, so $I' \models_b G_1^{N3}$.
We also have $I' \models_s G_1$, choosing $A_1(e) = I'(B_1(e))$.
Per our assumption, $I' \models_s G_2$:
there exists $A_2: E_2 → \Delta_{I'}$ such that $[I'+A_2](G_2)=\text{true}$.
By construction, all elements $r$ of $\Delta_{I'}$ are named, so there exists a function $\nu: \Delta_{I'} → T_G$ such that $I'(\nu(r)) = r$.
So we can build $B_2: E→T_G$ defined by $B_2(e) = \nu(A_2(e))$,
and therefore $I'[B_2](G_2^{N3})=\text{true}$.
It follows, by construction of $I'$, that $I[B2](G_2^{N3})=\text{true}$,
and therefore $I \models_b G_2^{N3}$.

$\Leftarrow$
Assume $G_1^{N3} \models_b G_2^{N3}$.
Let $I \models_s G_1$, i.e. there exists $A_1: C → \Delta_I$ such that $[I+A_1](G_1)=\text{true}$.
Without loss of generality, we assume that $IR$ does not contain any graph term.
We set $N_1 := \{e \in E_1 \,|\; \exists t \in C \cup L, D_I(t)=A_1(e)\}$~(the set of existentials that $A_1$ maps to a \emph{named} individual), and $O_1 := E_1 \setminus N_1$.
We can define an injection\footnote{
  For such an injection to exists, we must assume that $|O_1|$ is not greater than $|C|$. This is obviously true for finite graphs or countably infinite graphs, as $C$ contains an countably infinite amount of IRIs. We note that nothing in the definition of RDF constrains the set of blank nodes to be countable, so a graph using an uncountable amount of blank nodes is theoretically possible and would break this assumption.
  We consider, however, that this corner case has no practical relevance.
} $\delta$ which assigns to every resource in $A_1(O_1)$ a graph term of the form $\GTW{}{}{(c, c, c)}$ with $c \in C$.
We define the interpretation $I'$ as follows:
\begin{itemize}
\item $\Delta_{I'} = \Delta{I} \setminus A_1(O_1) \cup \delta(A_1(O_1))$
\item $D_{I'} = D_I$
\item $Q_{I'} =$ identity function on $\delta(A_1(O_1))$
\item $EXT_{I'} = \{(\delta^*(s),p,\delta^*(o)) \,|\; (s,p,o) \in EXT_I \}$ where $\delta^*(x) = \begin{cases}
\delta(r) & \text{if } r \in A_1(O_1)\\
r & \text{otherwise}
\end{cases}$
\end{itemize}
In other words, $I'$ is obtained by replacing in $I$ all resources $r \in A_1(O_1)$ by $\delta(r)$.
Clearly, $[I'+\delta^* \circ A_1](G_1) = \text{true}$, so $I' \models_s G_1$.
We also have $I' \models_b G_1^{N3}$ by choosing \[
B_1(e) = \begin{cases}
t \in C \cup L \text{ such that } D_I(t) = A_1(e) & \text{if } e \in N_1\\
\delta(A_1(e)) & \text{if } e \in O_1
\end{cases}
\]
Per our assumption, $I' \models_b G_2^{N3}$:
there exists $B_2: E_2 → T_G$ such that $I'[B_2](G_2^{N3}) = \text{true}$.
We can build $A_2: E_2 → IR$ as follows:
\[
A_2(e) = \begin{cases}
D_{I'}(B_2(e)) & \text{if } B_2(e) \in C \cup L\\
B_2(e) & \text{if } B_2(e) \text{is a graph term}
\end{cases}
\]
$[I+A_2](G_2) = \text{true}$ by construction, and therefore $I \models_s G_2$.
\end{proof}

\subsection{Semantic relations between graphs}\label{ssec:relations}
Having introduced basic models, we now discuss the semantic implications of the  syntactic relations between graphs we introduced in Sections~\ref{isomorphisms} and~\ref{other_graph_relationships}. We start with isomorphic graphs and show:

\begin{lemma}[Isomorphic graphs]\label{lemma:isographs}
If two closed graphs are isomorphic to each other, they share the same models.
\end{lemma}
\begin{proof}
Given two closed graphs $G_1=(U_1, E_1, F_1)$ and $G_2=(U_2, E_2, F_2)$ where $G_1$ is isomorphic to $G_2$ under $M$. Let $I\models_b G_1$.

If $G_1$ and $G_2$ are ground~($U_1 = U_2 = E_1 = E_2 = \emptyset$), we simply need to show that $I(F_2) = I(F_1)$. This holds because $F_1$ is isomorphic to $F_2$ under $M$, and $Q_I$ maps any two isomorphic ground graph terms to the same image.

If $G_1$ and $G_2$ are not ground,
we show that for each mapping $B_2:U_2\rightarrow T_G$ there exists a mapping $C_2:E_2\rightarrow T_G$ such that $I[B_2\bullet C_2](F_2)=true$. 
As $I\models_b G_1$, we know that for $B_1:=B_2\circ M:U_1\rightarrow T_G$
there exists a mapping $C_1:E_1\rightarrow T_G$ such that $I[B_1\bullet C_1](F_1)=true$. We set $C_2:= C_1\circ M^{-1}$ 
, then $I[B2\bullet C_2](F_2)=\text{true}$ and therefore $I \models_b G_2$.
\end{proof}

We can furthermore show, that N3 graphs b-entail their subgraphs.

\begin{lemma}[Subgraphs]
\label{lemma:subgraphs}
A closed N3 graph b-entails any of its subgraphs.
\end{lemma}
\begin{proof}
Let $G_1=(U_1,E_1,F_1)$ be a closed N3 graph, and $G_2=(U_2,E_2, F_2)$ a subgraph of $G_1$.
We notice that $G_2$ is necessarily closed~(as $FV(G_2) \subseteq FV(G_1) = \emptyset$) and therefore $FV(F_2) \subseteq U_2 \cup E_2$.
We also notice that, if $B: U_1 → T_G$ and $C: E_1 → T_G$, $I[B\bullet C](F2) = I[B|_{U_2}\bullet C|_{E_2}](F2)$.
Let $I$ be an interpretation such that $I \models_b G_1$. We show that for all $B_2: U_2 → T_G$, there exists $C_2: E_2 → T_G$ such that $I[B_2\bullet E_2](F_2) = \text{true}$.
As $I \models_b G_1$, we know that for $B_1: U1 → T_G$ such that $B_1|_{U_2} = B_2$, there exists $C_1: E_1 → T_G$ such that $I[B_1\bullet C_1](F_1)=\text{true}$.
It also holds that $I[B_1\bullet C_1](F_2)=\text{true}$, as $F_2 \subseteq F_1$,
and so $I[B_1|_{U_2}\bullet C_1|_{E_2}](F_2)=\text{true}$.
We set $C_2 := C_1|_{E_2}$, then $I[B_2\bullet C_2](F_2)=\text{true}$ and therefore $I \models_b G_2$.
\end{proof}

Similarly to RDF simple entailment, we see that E-instances entail their generalisations.

\begin{lemma}[E-Instances]
\label{lemma:e-instances}
A closed N3 graph is b-entailed by any of its E-instances.
\end{lemma}
\begin{proof}
Let $G_1 = (U, E_1, F_1)$ be a closed N3 graph, and $G_2 = (U, E_2, F_2)$ an E-instance of $G_1$ with the instantiation mapping $M: E_1 → T_G \cup E_2$.
Let $I \models_b G_2$. We show that for all $B: U → T_G$, there exists $C_1: E_1 → T_G$ such that $I[B\bullet C_1](F_1) = \text{true}$. As $I \models_b G_2$, we know that there exists $C_2: E_2 → T_G$ such that $I[B\bullet C_2](F_2) = \text{true}$.
We set
\[
C_1(v):=\begin{cases}
	C_2(M(v)) \text{ if } M(v) \in E_2\\
	M(v) \text{ if } M(v) \in T_G
\end{cases}
\]
We have
\TabPositions{0.5cm,4.5cm}
\begin{itemize}
\item[] 
   \tab $I[B\bullet C_2](F_2) = \text{true}$
\item[] = \tab $I[B\bullet C_2](M^t(F_1))$  \tab as $F_2 = M^t(F_1)$ per Definition~\ref{def:instance}
\item[] = \tab $I[B\bullet C_1](F_1)$       \tab by construction of $C_1$
 \end{itemize}
\end{proof}

For U-instances, the opposite is true: U-generalisations entail their instances.

\begin{lemma}[U-Instances]
\label{lemma:u-instances}
A closed N3 graph entailes any of its U-instances.
\end{lemma}
\begin{proof}
Let $G_1 = (U_1, E, F_1)$ be a closed N3 graph, and $G_2 = (U_2, E, F_2)$ a U-instance of $G_1$ with the instantiation mapping $M: U_1 → T_G \cup U_2$.
Let $I \models_b G_1$. We show that for all $B_2: U_2 → T_G$, there exists $C: E → T_G$ such that $I[B_2\bullet C](F_2) = \text{true}$. As $I \models_b G_1$, we know that for all $B_1: U_1 → T_G$ there exists $C: E → T_G$ such that $I[B_1\bullet C_](F_1) = \text{true}$.
We chose $B_1$ such that
\[
B_1(v) = \begin{cases}
	B_2(M(v)) \text{ if } M(v) \in U_2\\
	M(v) \text{ if } M(v) \in T_G
\end{cases}
\]
We have
\TabPositions{0.5cm,4.5cm}
\begin{itemize}
\item[]   \tab $I[B_1\bullet C](F_1) = \text{true}$
\item[] = \tab $I[B_2\bullet C](M^t(F_1))$  \tab by construction of $B_1$
\item[] = \tab $I[B_2\bullet C_1](F_2)$     \tab as $F_2 = M^t(F_1)$ per Definition~\ref{def:instance}
\end{itemize}
\end{proof}

The merge of two logical consequences is furthermore a logical consequence.

\begin{lemma}[Merge]\label{lemma:merge}
Let $G,G_1,G_2$ be closed N3 graphs, let $G_{m}$ be a merge of $G_1$ and $G_2$, and let $G\models_b G_1$ and $G\models_b G_2$ then
$G\models_b G_{m}$.
\end{lemma}
\begin{proof}
Let $G,G_1,G_2, G_m$ be as stated above, and let $I\models_b G$. 
Then $I\models_b G_1$ and $I\models_b G_2$. As $G_m$ is a merge of $G_1$ and $G_2$, there exist graphs $G_1'=(U_1', E_1', F_1')\simeq G_1$ and $G_2'=(U_2',E_2',F_2')\simeq G_2$ such  that $G_m=G_1'\sqcup G_2'$~(Definition~\ref{def:merge}). By Lemma~\ref{lemma:isographs} it holds that $I\models_b G_1'$ and $I\models_b G_2'$ and as $U_1'\cap U_2'=E_1'\cap E_2'=\emptyset$, this implies that  $I\models_b G_1'\sqcup G_2'$ and thus $I\models_b G_m$.
\end{proof}

\subsection{Interpolation Lemma}\label{sec:interpolation-lemma}
In the previous section we showed how several syntactic relations between graphs ensure entailment but we did not yet characterise the other direction: the syntactical relation between graphs which entail each other.
The goal of this section is to find an N3 counterpart for RDF's interpolation lemma. 

To do so, we first remind that, according to Lemma~\ref{lemma:simpleEntailment} that for pure RDF graphs b-entailmant coincides with RDF simple entailment.
In N3, we furthermore need to take universal quantification into account.
The graph 
\begin{equation}\label{eq:uniGT}
G_1=\G{x}{}{(x,\textit{:p},\textit{:o})}\end{equation}
b-entails
\begin{equation}G_2=\G{}{}{(\textit{:s},\textit{:p},\textit{:o})}\end{equation}
We replaced the universal variable $x$ by the ground term $\textit{:s}$ and retrieved a U-instance~(Lemma~\ref{lemma:u-instances}). This form of syntactical derivation gets more complicated if there are existentially quantified variables depending on the universals like for example in the following case:
\begin{equation}\label{eq:unGT1}
G_3=\G{x}{y}{(x,\textit{:hasMother},y)}
\end{equation}
Here, we state that everyone~(and everything) has a mother. We can safely b-entail that
\begin{equation}\label{eq:unGT2}
G_4=\G{}{y}{(\textit{:bob},\textit{:hasMother},y)}
\end{equation}
but the entailment
\begin{equation}\label{eq:unGT3}
G_4=\G{}{y}{(\textit{:bob},\textit{:hasMother},y), (\textit{:liz},\textit{:hasMother},y)}
\end{equation}
is incorrect and a correct version of this kind of entailment would be
\begin{equation}\label{eq:unGT4}
G_4=\G{}{y,z}{(\textit{:bob},\textit{:hasMother},y), (\textit{:liz},\textit{:hasMother},z)}
\end{equation}
to make sure that \textit{:bob} and \textit{:liz} possibly have different mothers. As we quantify over all ground terms and thus over an infinite set -- even if the set of constants is limited, we can construct infinitely many graph terms -- we cannot easily identify a finite graph from which all possible consequences can be derived, as it is the case for simple entailment in RDF where we have RDF's interpolation lemma.
Instead, we construct the \emph{universal extension} of a graph which might be infinite.
%
%
The idea is motivated by the concept of universal models for existential rules~\cite{universalModel}.

\begin{definition}[Universal Extension]
\label{def:universal-extension}
	Let $G=(U,E,F)$  be a finite and normalised closed N3 graph, $|U|=n$.  Let $N\subset V$ be a set of variables such that  $N\cap  C(G) =\emptyset$.
	We enumerate the variables in $U=\{u_1, u_2,\ldots u_n\}$ and define for each combination of ground terms 	$(t_1, \ldots, t_n)=\overline{t}\in T_G^n$  two mappings 
	$A_{\overline{t}}:U\rightarrow T_G$ 
	and $S_{\overline{t}}: E\rightarrow N$, such that:
	\begin{enumerate}
	\item[(i)] $A_{\overline{t}}(u_i)=t_i$ for all $1\leq i\leq n$,
		\item[(ii)] $S_{\overline{t}}(y_1)\neq S_{\overline{t'}}(y_2)$ if $\overline{t}\neq \overline{t'}$ or $y_1\neq y_2$
	\end{enumerate}
If $N$ is minimal in the sense that for each $x\in N$ there exists a value $\overline{t}\in T_G^{n}$ and a value $y \in E$ such that $S_{\overline{t}}(y)=x$, 
we call the graph $U_N(G)=(\emptyset,N,F')$ with 
	$F':=\bigcup_{\overline{t}\in T_G^n} (A_{\overline{t}}\bullet S_{\overline{t}})^t(F)$ a \emph{universal extension} of $G$ via $N$. In cases where the exact choice of $N$ is not important, we omit $N$ and call $U(G)$ a universal extension of $G$.
\end{definition}

Note that, modulo an appropriate mapping $B: N → T_G$, the universal extension contains all triples which must be true in any model of $G$. This is captured by the following lemma:

\begin{lemma}[Universal Extension equivalence]\label{lemma:uex}
Let $G = (U, E, F)$ be a finite normalized closed N3 graph,
and $U_N(G) = (\emptyset,N,F')$ a universal extension of $G$. Let $I$ be an 
interpretation, then the following holds:
\[I \models_b G \quad \text{ iff } \quad I\models_b U_N(G) \]
\end{lemma}
\begin{proof}

$\Rightarrow$
Let $I$ be a model of $G$.
We know that for all mappings $A_{\overline{t}}: U → T_G$, there exists a mapping $B_{\overline{t}}: E → T_G$ such that $I[A_{\overline{t}} \bullet B_{\overline{t}}](F)$ is true.
We also know~(from Definition~\ref{def:universal-extension}) that there exists a mapping $S_{\overline{t}}$ such that $(A_{\overline{t}} \bullet S_{\overline{t}})^t(F) \subseteq F'$. 
We also know that, for every $x \in N$, there exists $\overline{t} \in T_G^n$ and $y \in E$ such that $S_{\overline{t}}(y) = x$.
We set $B := B_{\overline{t}} \circ S_{\overline{t}}^{-1}$.
It follows that $I[B](F')$ is true, thus $I \models_b U_N(G)$.
 
$\Leftarrow$
Let $I$ be a model of $U_N(G)$.
We know that there exists $B: N → T_G$ such that $I[B](F')$ is true.
We also know~(from Definition~\ref{def:universal-extension}) that for all mapping $A_{\overline{t}}: U → T_G$ there exists a mapping $S_{\overline{t}}$ such that $(A_{\overline{t}} \bullet S_{\overline{t}})^t(F) \subseteq F'$.
We set $B_{\overline{t}} := B \circ S_{\overline{t}}$.
Then $I[A_{\overline{t}} \bullet B_{\overline{t}}](F)$ is true, thus $I \models_b G$.
\end{proof}

With universal extensions, we can provide the N3 version of the interpolation lemma for basic entailment:

\begin{lemma}[Generalised interpolation lemma]\label{lemma:interpolation}
		Given two finite normalised closed graphs $G_1=(U_1,E_1,F_1)$ and $G_2=(U_2,E_2,F_2)$ and universal extensions  $U_{N_1}(G_1)=(\emptyset, N_1, F'_1)$ of $G_1$  and $U_{N_2}(G_2)=(\emptyset, N_2, F'_2)$ of $G_2$  
        then the following holds:		
	\[G_1 \models_b G_2 \quad\text{ iff }\quad \parbox[t]{0.8\textwidth}{there exists a mapping $M: N_2 → N_1$ such that, for every triple $(s_2, p_2, o_2)$ of an E-instance of $U_{N_2}(G_2)$ there is some triple $(s_1, p_1, o_1)$ of $U_{N_1}(G_1)$ such that $s_2$, $p_2$, and $o_2$, are isomorphic under $M$ to $s_1$, $p_1$, and $o_1$ respectively.} \].
\end{lemma}

Before proving the lemma, we demonstrate how it relates to the RDF interpolation lemma for simple entailment~\cite[Section 5.3]{RDFsemantics}
with the following corollary:


\begin{corollary}[RDF interpolation lemma]\label{RDFint}
Given two finite normalised closed graphs $G_1=(\emptyset,E_1,F_1)$ and $G_2=(\emptyset,E_2,F_2)$ which do not contain graph terms.
$G_1$ b-entails a graph $G_2$ if and only if a subgraph of $G_1$ is an E-instance of $G_2$.
\end{corollary}
To see how the corollary is a special case of the lemma we first we take a closer look at the right hand side of the lemma.  Note  that universal extensions $U_{N_1}(G_1)$ and $U_{N_2}(G_2)$ are variants of $G_1$ and $G_2$ in which only the variables in $E_1$ and $E_2$ are renamed. As we exclude graph terms as components, the isomorphism $M$ between the elements $s_2$, $p_2$, $o_2$ and $s_1$, $p_1$, $o_1$ is only an isomorphism in the RDF sense. 

We come back to Lemma~\ref{lemma:interpolation} and provide its proof:
\begin{proof}
$\Rightarrow$ Let $G_1\models_b G_2$. 
We construct a model $I_U$ for $U_{N_1}(G_1)$.
We set $\Delta_{I_U}=T_G$, $D_{I_U}=\mathit{id}$ where $\mathit{id}$ stands for the identity function. We choose a representation of the isomorphic graph terms and  let $Q_{I_U}$ map graph terms to their isomorphic image. For each $n\in N_1$ we choose a unique IRI $B_1(n)\not\in C(G_1)\cup C(G_2)$. We set $EXT_{I_U}:= \{(I_U[B_1](s_1),I_U[B_1](p_1), I_U[B_1](o_1))|(s_1, p_1, o_1)\in F_1'\}$. We clearly see that $I_U\models_b U_{N_1}(G_1)$. By Lemma~\ref{lemma:uex}, it follows that $I_U\models_b G_1$, which by our assumption means that $I_U\models_b G_2$ and thus, again by Lemma~\ref{lemma:uex}, $I_U\models_b U_{N_2}(G_2)$.  That is, there is a mapping $B_2:N_2\rightarrow T_G$ such that $I_U[B_2]\models_b F_2'$, or  written differently $(I_U[B_2](s_2), I_U[B_2](p_2), I_U[B_2](o_2))\in EXT_{I_U}$ for each triple $(s_2,p_2,o_2)\in F_2'$.
We can therefore always find a triple $(s_1, p_1, o_1)\in F_1'$ such that
$(I_U[B_1](s_1), I_U[B_1](p_1), I_U[B_1](o_1))=(I_U[B_2](s_2), I_U[B_2](p_2), I_U[B_2](o_2))$.  We apply the mapping 
\[B_2'(x):= \begin{cases}B_2(x) &\text{ if } B_2(x)\in C(G_1)\\ x &\text{ else.} \end{cases}\]
to create an E-instance 
of $U_{N_2}(G_2)$. The claim follows for the mapping 
 \[M(x):= \begin{cases}B_1^{-1}\circ B_2(x) &\text{ if the value exists, } \\ n_1  &\text{ else.}\end{cases}\]
where $n_1\in N_1$ is an arbitrary value.

$\Leftarrow$ 
Let $M:N_2\rightarrow N_1$ be a mapping such that, for every triple $(s_2,p_2,o_2)$ of an E-instance $G_2'$ of $U_{N_2}(G_2)$ there exists some triple $(s_1,p_1,o_1)$ of $U_{N_1}(G_1)$ such that $s_2$, $p_2$, and $o_2$ are isomorphic under $M$ to $s_1$, $p_1$, and $o_1$, respectively. 
Let $I\models_b U_{N_1}(G_1)$. 
As $G_2'$ is an E-instance of  $U_{N_2}(G_2)$, there exists a mapping $W:X\subseteq N_2\rightarrow T_G\cup N_2$ such that  $G'_2=(\emptyset, N_2, W^t(F_2'))^N$.  
We thus have that $(\emptyset, N_1, (M^t(W^t(F_2')))$ is isomorphic to a subgraph of $U_{N_1}(G_1)$. 
It follows by Lemma~\ref{lemma:subgraphs} and Lemma~\ref{lemma:isographs}, that $I\models_b (\emptyset, N_1, (M^t(W^t(F_2')))$, that is, there exists a mapping $B:N_1\rightarrow T_G$, such that $I[B]((M^t(W^t(F_2'))))=\textit{true}$. 
We extend $W$ and $M$ to total mappings $W'$ and $M'$ on $N_2\cup T_G$ which behave like the identity function in cases where $W(x)$, respectively $M(x)$ are not defined.
Then $I[B \circ M'\circ W'](F_2')=\textit{true}$  holds
and thus $I\models_b U_{N_2}(G_2)$. The claim follows by Lemma~\ref{lemma:uex}.
\end{proof}


Note that Corollary \ref{RDFint} differs from RDF's interpolation lemma as stated in the specification 
in one detail:
RDF allows the graphs to be infinite while we only talk about finite graphs.
The reason is 
that is that our proof relies on finding "fresh" ground terms the existential variables can map to (variables need to correspond to ground terms). This is not possible for a graph $G$ using all possible ground terms, i.e. $C(G)=T_G$. For such a graph these terms always carry additional meaning.
We aim to address this issue in future work.

\subsection{N3-entailment}\label{n3-entailment}
In the previous sections we introduced and discussed N3's basic-entailment which -- as we explained -- is very similar to RDF's simple entailment. 
On top of this, we define N3-entailment which covers the N3-predicates collected under the \texttt{log}-namespace\footnote{\url{http://www.w3.org/2000/10/swap/log\#}}.
We introduce the predicates \texttt{log:implies} and \texttt{log:includes}. 
%
%
%
%

\begin{definition}[N3 models and entailment]
\label{def:n3Models}
\label{def:n3Entailment}
Let $I$ be an interpretation and $G$ and $H$ be graphs.
We call $I$ an N3 model of $G$, noted $I \models_{n3} G$, iff $I \models_{b} G$ and $I$ satisfies the following conditions:
\begin{itemize}
\item $I$ satisfies the the conditions for \emph{\texttt{log:implies}}~(see Definition~\ref{def:conditions_log_implies})
\item $I$ satisfies the the conditions for \emph{\texttt{log:includes}}~(see Definition~\ref{def:conditions_log_includes})
\end{itemize}
We say that $G$ N3-entails $H$, written as $G \models_{n3} H$ iff for each interpretation $I$ it holds that 
if $I \models_{n3} G$ then $I \models_{n3} H$.
\end{definition}

Note, that the \texttt{log}-namespace covers a wider range of predicates supporting, for example operations on literals~(\texttt{log:dtlit}, \texttt{log:langlit})  or types of terms~(\texttt{log:rawType}). We omit those here as the primary focus of this paper are graph terms. To take into account the missing cases, the above definition can be extended.

Below, we provide the definitions for the two predicates.
We start with the predicate \texttt{log:implies}, which is used for implications.


\begin{definition}[Semantic conditions for \texttt{log:implies}]
\label{def:conditions_log_implies}
The semantic conditions on an interpretation $I$ for \emph{\texttt{log:implies}} are:
\begin{itemize}
\item for all $s, o \in \Delta_I$,
$(s, D_I(\text{log:implies}), o) \in EXT_I$ holds only if either
\begin{itemize}
\item $s = \GTW{G_s}$ and $o = \GTW{G_o}$ are both graph terms and $I \models_{n3} G_o$ holds if $I \models_{n3} G_s$
\item $s$ or $g$ is not a graph term
\end{itemize}
\end{itemize}
\end{definition}
To illustrate how this definition makes it possible to apply rules on data, we come back to our example from the previous section:  
%
\begin{align}\tag{\ref{eq:tom}}
	&\G{}{}{(\text{:tom}, \text{rdf:type}, \text{:Human})}
\\
&\tag{\ref{eq:humanMotherInst}}
	\G{}{}{(
		\GT{}{}{(\text{:tom}, \text{rdf:type}, \text{:Human})},
		\text{log:implies},
		\GT{}{y}{(\text{:tom}, \text{:hasMother}, y)}  
		)}
\end{align}
Any model $I$ of these two triples  must satisfy:
\begin{itemize}
\item $(D_I(\text{:tom}), D_I(\text{rdf:type}), D_I(\text{:Human})) \in EXT_I$
\item $(
	D_I(\GT{}{}{(\text{:tom}, \text{rdf:type},\text{:Human})}),
	D_I(\text{log:implies}),
D_I(\GT{}{y}{(\text{:tom}, \text{:hasMother},\text{:y})}
)) \in EXT_I$
\end{itemize}
By the semantic condition on \texttt{log:implies}~(first bullet), and the fact that $I$ is a model for the graph on the left hand side~(``tom type human''), we derive that 
$I$ is also a model for the graph on the right-hand side.
The graph above thus N3-entails 
\begin{equation}\tag{\ref{eq:tomsMother}}
	\G{}{y}{(\text{:tom}, \text{:hasMother}, y)}.  
\end{equation}

If we were to replace \text{:tom} in the second line with a universally quantified variable $x$~(as in Formula~\ref{eq:humanMother}), the implication still holds. First, consider the implication as a ``regular'' assertion containing a universal. This assertion must hold for every possible value of \texttt{?x}, including \text{:tom}, so the second bullet above would still hold in every basic model of the graph. This would in turn lead to the N3-entailment that $\G{}{y}{(\text{:tom}, \text{:hasMother}, y)}.$, as above. 

The example also shows that N3-entailment supports a form of existential rules. Thus, in general, it is not decidable whether one graph is N3-entailed by  another~\cite{walking}. Other rule frameworks for RDF either exclude blank nodes on rule heads~(e.g. SWRL~\cite{swrl}) or limit how often such rules can be applied~(e.g. SPARQL RL~\cite{srl}). As blank nodes are 
widely used and important in different contexts, N3 allows them in rule heads. Some reasoners like for example EYE~\cite{eyepaper} provide mechanisms to tame such rules by different mechanisms.



The above mechanism even allows us to cover the "unquoting"
example from Section~\ref{motivation} which had a variable as rule head:
\begin{align}\tag{\ref{eq:firstGT}}
&	\texttt{:bob :says \{:alice :likes :pizza, :pasta\}. }\\
\tag{\ref{eq:weirdrule}}
&	\texttt{\{:bob :says ?x \} => ?x. }
\end{align}
Represented in abstract syntax, these formulae can be written as:
\begin{align} \label{eq:gt}
		  &\G{}{}{(\text{:bob}, \text{:says}, 	  \GT{}{}{(\text{:alice}, \text{:likes}, \text{:pizza}),(\text{:alice}, \text{:likes}, \text{:pasta})})}\\\label{eq:weirdruleAbs}
        & 	\G{x}{}{(
		\GT{}{}{(\text{:bob}, \text{:says}, x)},
		\text{log:implies},
		x
		)}
\end{align}
By simple entailment, we obtain Formula~\ref{eq:weirdruleinst} which we represent in abstract semantics as:
\begin{multline}
\G{}{}{(
		\GT{}{}{(\text{:bob}, \text{:says}, \GT{}{}{(\text{:alice}, \text{:likes}, \text{:pizza}),(\text{:alice}, \text{:likes}, \text{:pasta})})},
		\text{log:implies},\\
		\GT{}{}{(\text{:alice}, \text{:likes}, \text{:pizza}),(\text{:alice}, \text{:likes}, \text{:pasta})}
		)}
\end{multline}
With the condition for \text{log:implies}, we derive from Formulae~\ref{eq:gt} and~\ref{eq:weirdruleAbs}: 
\begin{equation}
		\G{}{}{(\text{:alice}, \text{:likes}, \text{:pizza}),(\text{:alice}, \text{:likes}, \text{:pasta})}
\end{equation}
which is the abstract representation of Formula~\ref{eq:res}.


Note that the semantic condition above is a one-way implication~(``only if'').
This means that given two graphs $G$ and $H$ such that $G \models_{n3} H$, it does not necessarily hold that $\GTW{G}$ \texttt{log:implies} $\GTW{H}$.
Note also that the semantic condition allows for \texttt{log:implies} to hold between a non-graph term and any other term~(last sub-bullet in Definition~\ref{def:conditions_log_implies}), in which case, it does not have any special meaning.




            




\vspace{1em}

Similarly to the conditions shown above, we define the semantics for the predicate \texttt{log:includes}. The predicate was initially introduced to mimic basic entailment~(more precisely its predecessor \emph{N3 derivation}~\cite[Section 5.7]{N3Logic}). It has two ground graph terms as arguments and is true if the second graph is b-entailed by the first one. 
A simple example of such a true triple~(in N3 syntax) is:
\begin{equation}\label{eq:logIncl}
\texttt{\{:s :p :o. :a :b :c\} log:includes \{:a :b :c\} }
\end{equation}
As the second argument of the predicate is a subgraph of the first we clearly see it is b-entailed by it~(see Lemma~\ref{lemma:subgraphs}). This example is easy to implement and supported by most reasoners we are aware of. 
However, b-entailment including universal quantification is rather complex as we illustrated earlier 
on Formulae~\ref{eq:unGT1}-\ref{eq:unGT4}.
%
To the best of our knowledge, there is furthermore no reasoner supporting it and
most practical use cases only require a simpler form of entailment which does not go beyond the expressivity of RDF's interpolation lemma 
(Corollary~\ref{RDFint})%
\footnote{Example uses can be found on the N3 test suite \url{https://codeberg.org/phochste/notation3tests}.}. We therefore exclude universal quantification here and define:


\begin{definition}[Existential consequence]\label{def:PC}
If a subgraph of a finite closed N3 graph $G_1$ is isomorphic to an E-instance of an N3 graph $G_2$, we call $G_2$ a \emph{existential consequence} of $G_1$.
\end{definition}
From our considerations in Section~\ref{ssec:relations}, we immediately see:

\begin{corollary}[Existential Consequence]\label{cor:PC}
An N3 graph b-entails any of its existential consequences.
\end{corollary}
\begin{proof}
This follows directly from Lemmas~\ref{lemma:isographs},~\ref{lemma:subgraphs} and~\ref{lemma:e-instances}. 
\end{proof}

Note that even graphs with universally quantified variables can have existential consequences. The concept only limits the kind of derivation steps taken into account and explicitly excludes consequences produced by instantiating universals. The graph

\begin{equation}
G=\G{x}{y}{(x,\textit{:p}, y)}\end{equation}
is an existential consequence of graph
$G_1$ in Formula~\ref{eq:uniGT}. 

We use the existential consequence to define the semantics of the predicate \texttt{log:includes}.

\begin{definition}[Semantic conditions for \texttt{log:includes}]
\label{def:conditions_log_includes}
The semantic conditions on an interpretation $I$ for \emph{\texttt{log:includes}} are:
\begin{itemize}
\item $(s, D_I(\text{log:includes}), o) \in EXT_I$ if and only if   $s = \GTW{G_s}$ and $o = \GTW{G_o}$ are both graph terms and 
$G_o$ is an existential consequence of $G_s$.
\end{itemize}
\end{definition}

Note that the semantic condition above is a two-way implication~(``if and only if'').
In other words, the extension of \texttt{log:includes} contains exactly all the pairs of graph terms present in $\Delta_I$ such that the second one is an existential consequence of the first one.



\section{Discussion of the requirements}\label{sec:eval}
Having defined the concept of basic semantics and N3 semantics, we now come back to the requirements identified in Section~\ref{sec:req} discussing them one by one.

\textbf{Semantic compatibility with RDF}
We defined basic entailment~(Definition~\ref{def:basicEntailment}) and characterised it through a generalised interpolation lemma~(Lemma~\ref{lemma:interpolation}). We furthermore showed that for RDF graphs b-entailment and simple entailment coincide~(Lemma~\ref{lemma:simpleEntailment}) and that the N3 interpolation lemma for b-entailment~(Lemma~\ref{lemma:interpolation}) is an extension of the RDF interpolation lemma for simple entailment on finite graphs~(Corollary~\ref{RDFint}).  
The main difference between the two is that the RDF interpolation lemma also holds for infinite graphs while the N3 interpolation lemma is limited to finite graphs.
As entailment between infinite RDF graphs is rarely needed in practice, we consider the difference between both forms of entailment minimal\footnote{
Furthermore, as explained in Section~\ref{sec:conclusion}, we conjecture that Lemma~\ref{lemma:interpolation} can be generalized to a broader class of graphs, not necessarily finite.
}. The semantics we defined is otherwise in line with RDF-entailment and the requirement is met.

\textbf{Referential Opacity}
The second requirement we identified was referential opacity. We meet this requirement by semantically representing graph terms as syntactic objects in base semantics which can then be further refined in N3 semantics -- or possible other extensions -- like we exemplified on implications and inclusion. We come back to our superman example~(Formulae~\ref{eq:super} and~\ref{eq:loislane}), 
which we represent in abstract syntax
\begin{align}
&\G{}{}{(\text{:superman}, \text{owl:sameAs}, \text{:clarkKent})}\\
&\G{}{}{(\text{:loisLane}, \text{:knows}, \GT{}{}{(\text{:superman}, \text{:can}, \text{:fly})}}
\end{align}
Even with an interpretation $I$ where $I(\text{:superman})=I(\text{:clarkKent})$, we have:
\begin{align*}
&I(\GT{}{}{(\text{:superman}, \text{:can}, \text{:fly})})=\GT{}{}{(\text{:superman}, \text{:can}, \text{:fly})}&\\
\neq \hspace{0.2cm}&I(\GT{}{}{(\text{:clarkKent}, \text{:can}, \text{:fly})}) =\GT{}{}{(\text{:clarkKent}, \text{:can}, \text{:fly})}& \end{align*}
With Rule~\ref{eq:transparent} which we write in abstract syntax as
\begin{equation}
\G{x,y,s,p,q,r}{}{(\GT{}{}{(x, \text{owl:sameAs}, y),(s,p,\GT{}{}{x,q,r})}, \text{log:implies},\GT{}{}{(s,p,\GT{}{}{(y,q,r)})})}
\end{equation}
This rule b-entails 
\begin{multline}
\G{}{}{(\GT{}{}{(\text{:superman}, \text{owl:sameAs}, \text{:clarkKent}),(\text{:loisLane},\text{:knows},\GT{}{}{(\text{:superman},\text{:can},\text{:fly})})},\\ \text{log:implies},
\GT{}{}{(\text{:loisLane},\text{:knows},\GT{}{}{(\text{:clarkKent},\text{:can},\text{:fly})})}}
\end{multline}
and with N3-entailment we get
\begin{equation}
\G{}{}{(\text{:loisLane},\text{:knows},\GT{}{}{(\text{:clarkKent},\text{:can},\text{:fly})})}
\end{equation}
Our semantics supports referential opacity, while allowing to \emph{locally} enable the derivations that would follow from referential transparency, as illustrated above.

\textbf{Closedness of graph terms}
Another requirement we identified was that graph terms in N3 need to be closed as they are used to represent premises and conclusions of rules. We again meet this criterion by interpreting graph terms by an isomorphic syntactical representation. In the case of implications, we further interpret the graph terms as graphs which in case they are true in a model can be used to derive new facts. If needed, IRIs could be used to represent open graphs. We could for example state
\[ \G{}{}{(\text{:loisLanesKnowledge}, \text{:contains}, \GT{}{}{(\text{:superman}, \text{:can}, \text{:fly})})}\]
and provide semantics for the predicate \textit{:contains} which ensure that the subject of the predicate represents a graph consisting of \emph{at least} the triples occurring in its object graphs. 
This idea has been discussed in previous work to represent RDF datasets using N3~\cite{N3Namedgraphs}. 

\textbf{Clarity and flexibility of quantification scopes}
We introduced explicit quantification to avoid possible ambiguities arising from implicit scoping.  For each graph we declare the set of universally and existentially quantified variables. This makes our semantics independent of the 
concrete syntax whose scoping of implicit quantification is still subject to discussion. Our semantics with explicit scoping makes it possible to better understand the impact of future scoping decisions and possible restrictions to the language.
In our approach, we chose to follow N3's original specification in which the scope of universal specification is always outside the scope of existential quantification when these are declared on the same graph~\cite[p.20]{N3Logic}. At least if we employ N3 entailment and not just basic entailment, this design decision does not restrict the expressivity of our logic. For example, if we want to modify our initial example on quantification~("Everybody loves somebody", Formula~\ref{eq:love}) to express that there exists someone who is loved by everyone, we can state 
\begin{equation}
\G{}{y}{(\GT{}{}{},\text{log:implies},\GT{x}{}{(x, \textit{:loves}, y)} )}
\end{equation}
According to our semantics, the empty graph is always true and in an N3 model the formula has thus the intended meaning.

\section{Related work}\label{sec:relatedwork}

The work presented in this paper has been influenced by several other publications. 

First, it relies on the original contributions introducing Notation3 Logic: W3C team submission~\cite{berners-lee_notation3_2011}, a web page by Tim Berners-Lee explaining design issues of the framework~\cite{designissues_2005}, and a research paper detailing the logic~\cite{N3Logic}. All N3 features discussed in this paper were informally introduced in these early specifications of the logic; however, none of them had formal semantics.
This explicitly includes graph terms and implicit quantification of two different types. 
N3 also offers the means to state 
 explicit quantification~(see also Formulae~\ref{eq:love3} and~\ref{eq:expli}). 
In the past, this feature often lead to problems because N3 treats the quantifiers themselves as parts of the graphs and 
thus as elements of sets. This makes it difficult to determine the proper scope
of a quantifier and their orders.
The specification allows us to write Formula~\ref{eq:love3} from above  
as
\begin{equation}\label{eq:love4}\tag{\ref*{eq:love}$'''$}
	\texttt{ @forSome :y. @forAll :x. :x :loves :y.}
\end{equation}
without changing its meaning (\emph{"everybody loves somebody"}).
As we consider this detail confusing,
we impose a fixed order for 
the two kinds of quantifiers.  Not only explicit quantification but also its implicit counterparts has lead to  contradicting interpretations.
In previous work, we investigated the extent and impact of these disagreements and proposed a mapping to a logic with explicit quantification~\cite{arndt2019implicit,arndt2015semantics}. This mapping could cope with 
formulae which are similar to First Order Logic like, for example, the implication detailed in Formula~\ref{eq:rule1}. But it did
not cover cases like the unquoting exemplified in Formula~\ref{eq:weirdrule}  or the built-in \texttt{log:includes}. 
The development of N3 was continued by a W3C community group which provided specifications for the syntax of the language~\cite{woensel2023n3}, its semantics~\cite{arndt2023n3semantics}, and its built-in functions~\cite{n3builtins}. 
The semantics developed in this context forms the base of the work we present in this paper. In addition to refining the definitions of that specification, we furthermore 
investigate N3's formal relationship to RDF simple entailment and add a semantic definition for the predicate \texttt{log:includes}. In earlier work this predicate relies on a simplified form of entailment called \emph{n3-derivation}~\cite[p.13]{N3Logic}. Even though n3-derivation was not fully detailed when introduced, we see a strong similarity to basic entailment as it was introduced in this paper. 

The idea of referencing to graphs, formulae, or statements has not only been discussed in the context of N3 but can also be found in different frameworks for the web and beyond. We first discuss web formalisms to then broaden our discussion  towards logic in general. In the moment of writing, the W3C RDF \& SPARQL working group is developing RDF 1.2~\cite{rdf12-concepts} and SPARQL~1.2~\cite{sparql-12} which add triple terms to RDF and SPARQL. Conceptually, these are very similar to graph terms with the main difference that users can not use whole graphs as terms of triples but only single triples. RDF 1.2 originates from RDF-star~(also RDF$^*$)~\cite{hartig2021rdf} which was originally proposed by Hartig~\cite{hartig2017foundations}.  In RDF 1.2 the use of triple terms is restricted and they are only allowed in the object position of a triple. Blank nodes in triple terms are globally quantified and triple terms are referentially transparent. In both of these aspects, triple terms differ from N3 graph terms. A singleton N3 graph term is not the same as a triple term in RDF 1.2. RDF does furthermore not support rules and it does not support any mechanism to unquote triple terms as illustrated for N3 with Example~\ref{eq:weirdrule}. 

Before the rise and standardisation of triple terms, RDF reification~\cite[Appendix D1]{RDFsemantics} was used to make statements about triples. There was no special syntactic construct to refer to triples but these were described in RDF's triple syntax. A reified triple was represented by an IRI or blank node
in combination with predicates making connections to the different components of a triple~(\texttt{rdf:subject}, \texttt{rdf:predicate}, \texttt{rdf:object}). There was no semantics for RDF reification and no guarantees or checks that a reified triple was well-formed~(we could for example have two subject declarations for the same term). However, by the way RDF reification is done, we can assume that it is referentially transparent and that all existential quantification is global.

As an alternative to RDF reification, Nguyen et al. 
proposed \emph{singleton properties}~\cite{nguyen2014don}.  The idea behind this concept is to use, in each triple to be reified, a fresh predicate that is related to the original property. That predicate can then be used to refer to the single triple where it is used. This approach to refer to triples does not only differ syntactically from the former two, it also differs in another sense: while triple terms and also reification, just like N3 graph terms, can be used without stating the triple, this is not possible with singleton properties. 
We can, for example, not express Formula~\ref{eq:super} without asserting that \texttt{:superman :can :fly}.

To refer to graphs RDF 1.1 furthermore introduced the concept of RDF datasets~\cite[Section 4]{rdf11-concepts}. Here, a combination of a graph name and a graph can be stated in an RDF document and the name can then further be used to refer to that construct. Although RDF datasets are part of the RDF 1.1 recommendation, their semantics has not been fixed and there exist different proposals~\cite{namedgraphs} addressing the different aspects identified in Section~\ref{sec:req} in various ways most of which can be represented by N3~\cite{N3Namedgraphs}. The semantics introduced by Carroll et al.~\cite{namedGraphsHayes} is the closest to the graph term semantics introduced in this paper. A named graph in this semantics is a pair composed of a name and a graph, and that name can be used in triples and always refers to the pair itself.
As N3 graph terms do not have names, we observe a difference here. The graph is closed, blank nodes are local, and the graph itself is interpreted as a syntactical object thereby supporting referential opacity. 

Outside of the web context, Common Logic~\cite{CommonLogic} defines constructs similar to RDF datasets and their semantics according to Carroll et al. Here, we  can again combine formulae with names to then use these names in other formulae. In contrast to RDF datasets, Common Logic supports universal quantification but it does not support to quantify into a formula. The expression
\begin{equation}
\G{x}{}{(x, \text{says}, \GT{}{}{(x, \text{is}, \text{wise})})}
\end{equation}
is supported by our semantics of N3 but cannot be expressed in Common Logic. The difficulty here is that the quantified variable $x$ occurs inside and outside of the graph term. We solve that by first grounding the variable, then further interpret it according to the position where it is used~(once as a direct term, once as part of the graph term). Common Logic does furthermore not support unquoting. 
In contrast to that, \emph{Qiana}, a framework developed by Coumes et al.~\cite{qiana}, supports both quantifying into a formula/graph term and unquoting. The mechanism to do so differs from our approach, as its formalisation treats quoted formulae as logical functions which directly take terms of the domain of discourse as arguments. This can be used to resolve quantification over variables used simultaneously inside and outside a quoted formula -- we use grounding instead. Qiana furthermore diverges from N3 in that its main focus is to display formulae together with some context under which they may or may not be true. In that regard, it is close to the different variants of modal logic~\cite{modalLogic}.  
We adopt Qiana's notation of underlining formulae/graph terms to indicate their mention rather than use.
 
The concept of universal extension of N3 graphs~(Definition~\ref{def:universal-extension}) and the related theory have been influenced by the work around \emph{existential rules}~\cite{walking}
(also called \emph{tuple-generating dependencies}~\cite{tgds} or 
\emph{Datalog\({}^{\mbox{{\(\pm\)}}}\)}~\cite{datalog+-}).  
Here, the \emph{chase} procedure can be used to construct \emph{universal models}~\cite{universalModel}: existential rules are applied repeatedly to a growing set of facts, where each rule application adds new facts unless its head already maps homomorphically into the existing ones. In the case of basic entailment, we do not have rules but $\forall\exists$-formulae which could be seen as existential rules with empty bodies. To mimic RDF's interpolation lemma, we therefore have to apply them on all terms producing infinite models. Existential rules are furthermore closely related to N3 rules employing the built-in predicate \texttt{log:implies}~\cite{exN3,arndt2023notation3} and query answering for N3 entailment faces the same problems as for existential rules and is in general undecidable.


\section{Conclusion and future work}\label{sec:conclusion}

We have proposed an abstract syntax and a formal semantics for N3.
Compared to previous efforts in that direction, our proposal aims to cover the different ways in which graph terms have been used historically in N3: statements about statements, and rules.
This formal semantics is defined in two levels: basic semantics (comparable to RDF's simple entailment), and N3 semantics, which includes N3's special predicates, notably the rule-defining predicate \texttt{log:implies}.
We have shown that this formal semantics satisfies the requirement we identified as meeting those different use-cases,
and that RDF's interpolation lemma for simple entailment can be extended to N3's basic semantics (on finite graphs).

While we proved our generalized interpolation lemma for finite graphs only, we conjecture that a proof may exist for infinite graphs as well. As a matter of fact, our proof does not prevent the graphs to have an infinite number of triples; it only requires them to not use all the variables (so that we can build a universal extension) nor all the ground terms (so that we can build a mapping $B_1$). But we believe that even those constrains could be lifted, and plan to address them in future work.

We also acknowledge that our generalized interpolation lemma is far less practical than its counterpart in RDF:
it relies in the universal extension of N3 graphs,
which are infinite as soon as the graph contains a universal variable
(unlike universal models~\cite{universalModel}, the concept of which universal extensions are inspired).
In the future, we wish to provide an entailment procedure which only takes the terms occurring in the two involved graphs into account 
similar to the
RDFS entailment rules provided in the RDF specification \cite{DBLP:journals/ws/Horst05,RDFsemantics}.

Finally, while we formally define in N3-entailment the semantics of the special predicates \texttt{log:implies} and \texttt{log:includes}, the original N3 specification (and its implementations) provide much more special predicates.
We plan to include them, following the example of the two we described in this paper.
Note that some of these additional properties will include new challenges,
such as negation (e.g. \texttt{log:notIncludes}\footnote{
\url{https://w3c-cg.github.io/N3/spec/builtins.html\#log:notIncludes}
})
and aggregation (e.g. \texttt{log:forAllIn}\footnote{
\url{https://w3c-cg.github.io/N3/spec/builtins.html\#log:forAllIn}
}). As N3 allows variables to occur in the predicate position of triples, 
these  concepts cannot rely on classical stratification \cite{DBLP:books/mk/minker88/AptBW88}. We plan to base our definitions on \emph{chain stratification}~\cite{KIAK2026}, a concept recently proposed to overcome that problem.



\appendix




\bibliography{bibliography}

@techreport{beckett_turtle_nodate,
	title        = {{RDF} 1.1 Turtle},
	author       = {Eric Prud'hommeaux and Gavin Carothers},
	year         = 2014,
	month        = feb,
	note         = {https://www.w3.org/TR/2014/REC-turtle-20140225/},
	type         = {{W3C} Recommendation},
	institution  = {{W3C}}
}

@techreport{owl,
	title        = {{OWL} 2 Web Ontology Language
Document Overview (Second Edition)},
	author       = {{W3C OWL Working Group}},
	year         = 2012,
	month        = dec,
	note         = {\url{https://www.w3.org/TR/owl2-overview/}},
	type         = {{W3C} Recommendation},
	institution  = {{W3C}}
}

@techreport{rdf11-concepts,
	title        = {{RDF} 1.1 Concepts and Abstract Syntax},
	author       = {Markus Lanthaler and David Wood and Richard Cyganiak},
	year         = 2014,
	month        = feb,
	url         = {https://www.w3.org/TR/2014/REC-rdf11-concepts-20140225/},
	type         = {{W3C} Recommendation},
	institution  = {{W3C}}
}

@techreport{rdf12-concepts,
	title        = {{RDF} 1.2 Concepts and Abstract Syntax},
	author       = {Gregg Kellogg and
Olaf Hartig and
Pierre-Antoine Champin and
Andy Seaborne},
	year         = 2026,
	month        = apr,
	url         = {https://www.w3.org/TR/rdf12-concepts/},
	type         = {{W3C} Candidate Recommendation Snapshot},
	institution  = {{W3C}}
}

@techreport{berners-lee_notation3_2011,
	title        = {Notation3 ({N3}): {A} readable {RDF} syntax},
	author       = {Berners-Lee, Tim and Connolly, Dan},
	year         = 2011,
	url          = {https://www.w3.org/TeamSubmission/n3/},
	urldate      = {2021-07-15},
	institution  = {{W3C}}
}

@techreport{designissues_2005,
	title        = {{Notation 3 Logic}},
	author       = {Berners-Lee, Tim},
	year         = 2005,
	url          = {https://www.w3.org/DesignIssues/Notation3},
	institution  = {{W3C}}
}

@article{N3Logic,
	title        = {{N3Logic: A logical framework for the World Wide Web}},
	author       = {Berners-Lee, Tim and Connolly, Dan and Kagal, Lalana and Scharf, Yosi and Hendler, Jim},
	year         = 2008,
	journal      = {Theory and Practice of Logic Programming},
	publisher    = {Cambridge University Press},
	volume       = 8,
	number       = 3,
	pages        = {249--269},
	doi          = {10.1017/S1471068407003213},
	issn         = {1471-0684},
	url          = {http://arxiv.org/abs/0711.1533},
	modificationdate = {2021-09-09T10:19:27}
}

@misc{cwm,
	title        = {{Cwm}},
	author       = {Berners-Lee, Tim},
	year         = {2000--2009},
	note         = {\url{http://www.w3.org/2000/10/swap/doc/cwm.html}}
}

@techreport{woensel2023n3,
	title        = {Notation3 Language},
	author       = {William {Van Woensel} and Dörthe Arndt and Pierre-Antoine Champin and Dominik Tomaszuk and Gregg Kellogg},
	year         = 2025,
	month        = oct,
	note         = {\url{https://w3c.github.io/N3/reports/20230703/}},
	type         = {{W3C} Community Group Report},
	institution  = {{W3C}}
}

@techreport{arndt2023n3semantics,
	title        = {Notation3 Semantics},
	author       = {Dörthe Arndt and Pierre-Antoine Champin},
	year         = 2025,
	month        = oct,
	note         = {\url{https://w3c.github.io/N3/reports/20230703/semantics.html}},
	type         = {{W3C} Community Group Report},
	institution  = {{W3C}}
}

@techreport{n3builtins,
	title        = {Notation3 Builtin Functions},
	author       = {William {Van Woensel} and Patrick Hochstenbach},
	year         = 2023,
	month        = jul,
	note         = {\url{https://w3c.github.io/N3/reports/20230703/builtins.html}},
	type         = {{W3C} Community Group Report},
	institution  = {{W3C}}
}

@inproceedings{arndt2015semantics,
	title        = {Semantics of {Notation3} logic: A solution for implicit quantification},
	author       = {Arndt, D{\"o}rthe and Verborgh, Ruben and De Roo, Jos and Sun, Hong and Mannens, Erik and Van de Walle, Rik},
	year         = 2015,
	booktitle    = {Rule Technologies: Foundations, Tools, and Applications: 9th International Symposium, RuleML 2015, Berlin, Germany, August 2-5, 2015, Proceedings 9},
	pages        = {127--143},
	organization = {Springer},
    doi          = {10.1007/978-3-319-21542-6_9}
}

@techreport{swrl,
	title        = {{SWRL}: A semantic web rule language combining {OWL} and {RuleML}},
	author       = {Horrocks, Ian and Patel-Schneider, Peter F and Boley, Harold and Tabet, Said and Grosof, Benjamin and Dean, Mike},
	year         = 2004,
	month        = may,
	note         = {\url{https://www.w3.org/Submission/2004/SUBM-SWRL-20040521/}},
	type         = {{W3C} Member Submission},
	institution  = {{W3C}}
}

@techreport{srl,
	title        = {{SPARQL 1.2 RL}},
	author       = {Robert David and
David Habgood and
Andy Seaborne and
Simon Steyskal },
	year         = 2026,
	month        = sep,
	note         = {\url{https://www.w3.org/TR/sparql12-rl/}},
	type         = {{W3C} Working Draft},
	institution  = {{W3C}}
}

@article{eyepaper,
	title        = {Drawing Conclusions from Linked Data on the Web: The {EYE} Reasoner},
	author       = {Verborgh, Ruben and De Roo, Jos},
	year         = 2015,
	month        = may,
	journal      = {IEEE Software},
	volume       = 32,
	number       = 5,
	pages        = {23--27},
    doi = {10.1109/MS.2015.63}
}

@techreport{sparql-12,
	title        = {{SPARQL} 1.2 Query Language},
	author       = {Olaf Hartig and
Andy Seaborne and
Ruben Taelman and
Gregory Williams and
Thomas Pellissier Tanon},
	year         = 2026,
	month        = sep,
	note         = {\url{https://www.w3.org/TR/sparql12-query/}},
	type         = {{W3C} Working Draft},
	institution  = {{W3C}}
}

@techreport{Prudhommeaux2014RT,
	title        = {{RDF} 1.1 Turtle},
	author       = {Eric Prud'hommeaux and Gavin Carothers},
	year         = 2014,
	month        = feb,
	url         = {https://www.w3.org/TR/2014/REC-turtle-20140225/},
	type         = {{W3C} Recommendation},
	institution  = {{W3C}}
}

@techreport{RDFsemantics,
	title        = {{RDF} 1.1 {Semantics}},
	author       = {Hayes, Pat and Patel-Schneider, Peter F.},
	year         = 2014,
	month        = feb,
	day          = 25,
	url          = {https://www.w3.org/TR/rdf11-semantics/},
	urldate      = {2023-01-04},
	type         = {{W3C} Recommendation},
    institution  = {{W3C}}
}

@article{kafer2018rule,
  title={Rule-based Programming of User Agents for Linked Data.},
  author={K{\"a}fer, Tobias and Harth, Andreas},
  journal={LDOW@ WWW},
  volume={2073},
  year={2018},
  url={https://ceur-ws.org/Vol-2073/article-05.pdf}
}

@article{arndt2019implicit,
  title={Implicit quantification made explicit: How to interpret blank nodes and universal variables in {Notation3} Logic},
  author={Arndt, D{\"o}rthe and Schrijvers, Tom and De Roo, Jos and Verborgh, Ruben},
  journal={Journal of Web Semantics},
  volume={58},
  pages={100501},
  year={2019},
  publisher={Elsevier},
  doi={10.1016/j.websem.2019.04.001}
}

@inproceedings{arndt2023notation3,
  title={Notation3 as an existential rule language},
  author={Arndt, D{\"o}rthe and Mennicke, Stephan},
  booktitle={International Joint Conference on Rules and Reasoning},
  pages={70--85},
  year={2023},
  organization={Springer},
  doi={10.1007/978-3-031-45072-3_5},
}

@article{exN3,
	author    = {D{\"{o}}rthe Arndt and Stephan Mennicke},
	title     = {Existential {Notation3} Logic},
	journal   = {Theory and Practice of Logic Programming},
	publisher = {Cambridge University Press},
	year      = {2025},
	month     = {June},
	pages     = {1-36},
    doi       = {10.1017/S1471068425000055}
}

@techreport{hartig2021rdf,
	title        = {{RDF-star} and {SPARQL-star}},
	author       = {Olaf Hartig and
Pierre-Antoine Champin  and
Gregg Kellogg and
Andy Seaborne },
	year         = 2021,
	month        = dec,
	note         = {\url{https://www. w3. org/2021/12/rdf-star. html}},
	type         = {{W3C} Community Group Report},
	institution  = {{W3C}}
}

@inproceedings{hartig2017foundations,
  title={Foundations of {RDF*} and {SPARQL*}:(An alternative approach to statement-level metadata in {RDF})},
  author={Hartig, Olaf},
  booktitle={AMW 2017 11th Alberto Mendelzon International Workshop on Foundations of Data Management and the Web, Montevideo, Uruguay, June 7-9, 2017.},
  volume={1912},
  year={2017},
  organization={Juan Reutter, Divesh Srivastava},
  url={https://ceur-ws.org/Vol-1912/paper12.pdf}
}

@inproceedings{nguyen2014don,
  title={Don't like {RDF} reification? {M}aking statements about statements using singleton property},
  author={Nguyen, Vinh and Bodenreider, Olivier and Sheth, Amit},
  booktitle={Proceedings of the 23rd international conference on World Wide Web},
  pages={759--770},
  year={2014},
  doi={10.1145/2566486.2567973}
}

@techreport{namedgraphs,
	title        = {{RDF 1.1}: On Semantics of {RDF} Datasets},
	author       = {Antoine Zimmermann},
	year         = 2014,
	month        = feb,
	url         = {https://www.w3.org/TR/rdf11-datasets/},
	type         = {{W3C} Working Group Note},
	institution  = {{W3C}}
}

@inproceedings{N3Namedgraphs,
  author       = {D{\"{o}}rthe Arndt and
                  William Van Woensel},
  editor       = {Ahmet Soylu and
                  Sotiris Moschoyiannis and
                  Guido Governatori and
                  Mantas Simkus and
                  Petros S. Stefaneas and
                  Alexander Steen and
                  Adrian Giurca},
  title        = {Towards Supporting Multiple Semantics of Named Graphs Using {N3} Rules},
  booktitle    = {Proceedings of the 13th RuleML+RR 2019 Doctoral Consortium and Rule
                  Challenge, September 16-19, 2019 - Bolzano, Italy},
  series       = {{CEUR} Workshop Proceedings},
  volume       = {2438},
  publisher    = {CEUR-WS.org},
  year         = {2019},
  url          = {https://ceur-ws.org/Vol-2438/paper6.pdf},
  bibsource    = {dblp computer science bibliography, https://dblp.org}
}

@inproceedings{namedGraphsHayes,
  author       = {Jeremy J. Carroll and
                  Christian Bizer and
                  Patrick J. Hayes and
                  Patrick Stickler},
  editor       = {Allan Ellis and
                  Tatsuya Hagino},
  title        = {Named graphs, provenance and trust},
  booktitle    = {Proceedings of the 14th international conference on World Wide Web,
                  {WWW} 2005, Chiba, Japan, May 10-14, 2005},
  pages        = {613--622},
  publisher    = {{ACM}},
  year         = {2005},
  url          = {https://doi.org/10.1145/1060745.1060835},
  doi          = {10.1145/1060745.1060835},
  bibsource    = {dblp computer science bibliography, https://dblp.org}
}

@inproceedings{universalModel,
  author       = {Alin Deutsch and
                  Alan Nash and
                  Jeffrey B. Remmel},
  editor       = {Maurizio Lenzerini and
                  Domenico Lembo},
  title        = {The chase revisited},
  booktitle    = {Proceedings of the Twenty-Seventh {ACM} {SIGMOD-SIGACT-SIGART} Symposium
                  on Principles of Database Systems, {PODS} 2008, June 9-11, 2008, Vancouver,
                  BC, Canada},
  pages        = {149--158},
  publisher    = {{ACM}},
  year         = {2008},
  url          = {https://doi.org/10.1145/1376916.1376938},
  doi          = {10.1145/1376916.1376938},
  bibsource    = {dblp computer science bibliography, https://dblp.org}
}

@article{walking,
  author       = {Jean{-}Fran{\c{c}}ois Baget and
                  Michel Lecl{\`{e}}re and
                  Marie{-}Laure Mugnier and
                  Eric Salvat},
  title        = {On rules with existential variables: Walking the decidability line},
  journal      = {Artif. Intell.},
  volume       = {175},
  number       = {9-10},
  pages        = {1620--1654},
  year         = {2011},
  doi          = {10.1016/J.ARTINT.2011.03.002},
  bibsource    = {dblp computer science bibliography, https://dblp.org}
}

@article{qiana,
  author       = {Simon Coumes and
                  Pierre{-}Henri Paris and
                  Fran{\c{c}}ois Schwarzentruber and
                  Fabian M. Suchanek},
  title        = {Qiana: {A} First-Order Formalism to Quantify over Contexts and Formulas
                  with Temporality},
  journal      = {J. Artif. Intell. Res.},
  volume       = {85},
  year         = {2026},
  doi          = {10.1613/JAIR.1.18402},
  bibsource    = {dblp computer science bibliography, https://dblp.org}
}

@techreport{CommonLogic,
  author       = {{ISO/IEC}},
  title        = {Information technology --- {Common Logic ({CL})}: 
                  A framework for a family of logic-based languages},
  type         = {Standard},
  number       = {ISO/IEC 24707:2018},
  edition      = {2},
  year         = {2018},
  institution  = {International Organization for Standardization},
  address      = {Geneva, Switzerland},
}

@book{modalLogic,
  author       = {Patrick Blackburn and
                  Maarten de Rijke and
                  Yde Venema},
  title        = {Modal Logic},
  series       = {Cambridge Tracts in Theoretical Computer Science},
  volume       = {53},
  publisher    = {Cambridge University Press},
  year         = {2001},
  url          = {https://doi.org/10.1017/CBO9781107050884},
  doi          = {10.1017/CBO9781107050884},
  isbn         = {978-1-10705088-4},
  bibsource    = {dblp computer science bibliography, https://dblp.org}
}

@inproceedings{datalog+-,
  author       = {Andrea Cal{\`{\i}} and
                  Georg Gottlob and
                  Thomas Lukasiewicz},
  editor       = {Ronald Fagin},
  title        = {Datalog\({}^{\mbox{{\(\pm\)}}}\): a unified approach to ontologies
                  and integrity constraints},
  booktitle    = {Database Theory - {ICDT} 2009, 12th International Conference, St.
                  Petersburg, Russia, March 23-25, 2009, Proceedings},
  series       = {{ACM} International Conference Proceeding Series},
  volume       = {361},
  pages        = {14--30},
  publisher    = {{ACM}},
  year         = {2009},
  url          = {https://doi.org/10.1145/1514894.1514897},
  doi          = {10.1145/1514894.1514897},
  bibsource    = {dblp computer science bibliography, https://dblp.org}
}

@book{tgds,
  author       = {Serge Abiteboul and
                  Richard Hull and
                  Victor Vianu},
  title        = {Foundations of Databases},
  publisher    = {Addison-Wesley},
  year         = {1995},
  url          = {http://webdam.inria.fr/Alice/},
  isbn         = {0-201-53771-0},
  bibsource    = {dblp computer science bibliography, https://dblp.org}
}

@inproceedings{KIAK2026,
  author    = {Nils K{\"{u}}chenmeister and Alex Ivliev and D{\"{o}}rthe Arndt
               and Markus Kr{\"{o}}tzsch},
  title     = {Stratified Negation in {RDF} Rules: A Correct Approach},
  editor    = {Manolis Koubarakis and Maria-Esther Vidal and Axel Polleres and
               Marieke van Erp and Ernesto Jimenez Ruiz and Oshani Seneviratne
               and Lora Aroyo and Gianluca Demartini and Reham Alharbi and
               Roberto Barile and Claudia {d'Amato} and Valentina Tamma},
  booktitle = {Proceedings of the 25th International Semantic Web Conference
               (ISWC 2026)},
  publisher = {Springer},
  year      = {2026}
}

@InProceedings{Genesereth2015,
author="Genesereth, Michael
and Kao, Eric",
editor="Bassiliades, Nick
and Gottlob, Georg
and Sadri, Fariba
and Paschke, Adrian
and Roman, Dumitru",
title={The {Herbrand} Manifesto},
booktitle="Rule Technologies: Foundations, Tools, and Applications",
year="2015",
publisher="Springer International Publishing",
address="Cham",
pages="3--12",
doi={10.1007/978-3-319-21542-6_1},
isbn="978-3-319-21542-6"
}

@article{DBLP:journals/ws/Horst05,
  author       = {Herman J. ter Horst},
  title        = {Completeness, decidability and complexity of entailment for {RDF}
                  Schema and a semantic extension involving the {OWL} vocabulary},
  journal      = {J. Web Semant.},
  volume       = {3},
  number       = {2-3},
  pages        = {79--115},
  year         = {2005},
  doi          = {10.1016/J.WEBSEM.2005.06.001},
  bibsource    = {dblp computer science bibliography, https://dblp.org}
}

@incollection{DBLP:books/mk/minker88/AptBW88,
  author       = {Krzysztof R. Apt and
                  Howard A. Blair and
                  Adrian Walker},
  editor       = {Jack Minker},
  title        = {Towards a Theory of Declarative Knowledge},
  booktitle    = {Foundations of Deductive Databases and Logic Programming},
  pages        = {89--148},
  publisher    = {Morgan Kaufmann},
  year         = {1988},
  doi          = {10.1016/B978-0-934613-40-8.50006-3},
  bibsource    = {dblp computer science bibliography, https://dblp.org}
}

\end{document}